\documentclass[aps,prl,twocolumn,amsmath,amssymb,showpacs,superscriptaddress,notitlepage,longbibliography]{revtex4-2}
\usepackage{CJKutf8}
\usepackage[colorlinks=true,linkcolor=blue,anchorcolor=red,citecolor=blue, urlcolor=blue]{hyperref}
\usepackage{bm} 
\usepackage{bbm} 
\usepackage{graphicx}
\usepackage{graphicx}
\usepackage{adjustbox}
\usepackage{color}
\usepackage{xcolor}
\usepackage{soul}
\usepackage{enumerate}
\usepackage{mathrsfs}
\usepackage{tikz}
\usepackage{circuitikz}

\usetikzlibrary{patterns}
\usetikzlibrary{through,hobby}
\usetikzlibrary{decorations.markings}
\usetikzlibrary{fadings,shadings}
\usetikzlibrary{plotmarks}
\usetikzlibrary{positioning}
\usetikzlibrary{decorations,arrows}
\usetikzlibrary{decorations.pathreplacing}
\usetikzlibrary{chains, scopes, positioning, backgrounds, shapes, fit, shadows, calc, arrows.meta}
\usepackage{braket}

\newcommand{\tr}{\operatorname{Tr}}

\usepackage{amsthm}
\newtheorem{thm}{Theorem}
\newtheorem{dfn}[]{Definition}

\newtheorem{corollary}[thm]{Corollary}

\newcommand{\N}{\mathcal{N}}
\newcommand{\cl}{\rho_{\mathrm{cl}}}

\definecolor{myBlue}{HTML}{0e35b3}

\begin{document}

\begin{CJK*}{UTF8}{}
\title{Strong-to-weak spontaneous symmetry breaking of higher-form non-invertible symmetries in Kitaev's quantum double model}

\author{Zijian Song (\CJKfamily{bsmi}宋子健)}
\email{zjsong.physics@gmail.com}
\affiliation{C. N. Yang Institute for Theoretical Physics, Stony Brook University, Stony Brook, NY 11794, USA}

\author{Jian-Hao Zhang (\CJKfamily{gbsn}张健豪)}
\email{Sergio.Zhang@colorado.edu}
\affiliation{Department of Physics and Center for Theory of Quantum Matter, University of Colorado, Boulder, CO 80309, USA}

\begin{abstract}
Topological orders can be understood as spontaneous symmetry breaking of higher-form symmetries. In the non-abelian case, the broken higher-form symmetries are notably non-invertible. In this work, we extend this framework to mixed states, where symmetries can be either strong or weak. In particular, we investigate the strong-to-weak spontaneous symmetry breaking (SWSSB) of non-invertible higher-form symmetries in non-abelian Kitaev's quantum double models under decoherence. We further show that the resulting decohered quantum double mixed states form a locally indistinguishable set, which also constitutes an information convex set. Importantly, we emphasize that the number of extremal points of this convex set equals the ground-state degeneracy of the corresponding pure state, highlighting that the quantum information encoded in the ground-state subspace is degraded into classical information captured by the convex set of decohered density matrices.
\end{abstract}

\maketitle
\end{CJK*}

\textit{Introduction} -- Topological phases of matter with different types of symmetry have been extensively studied over the past few decades \cite{Kitaev_2003, Kitaev_2006, Wen_2015, Wen_2017, Johnson_Freyd_2022}, including the stable ground-state degeneracy, long-range entanglement, quasiparticle excitations with exotic braiding statistics, etc. In particular, the stability of topological order (TO) under local perturbations highlights its potential for advancing fault-tolerant quantum computation \cite{Kitaev_2006, freedman2000modularfunctoruniversalquantum, TQM}. Recently, the concept of symmetry has been generalized in many senses: On the one hand, a distinctive feature of mixed quantum states is the distinction between strong (exact) symmetry and weak (average) symmetry. On the other hand, recent developments in quantum field theory reveal that the mathematical structure of symmetry is not limited to groups. Symmetry operators can be generalized to topological operators in spacetime, which may be \textit{non-invertible}. This extends the concept of symmetry from groups to fusion categories~\cite{schafernameki2023ictp, shao2024tasi, Luo_2024, thorngren2019fusion, thorngren2021fusion, Choi_2022, Bhardwaj_2022, Zhang2023anomaly, Cordova_2024}. Furthermore, even topological quantum phases once thought to lie beyond Landau’s symmetry-breaking classification have been incorporated into the traditional paradigm through a generalized notion of symmetry \cite{McGreevy_2023, Gaiotto_2015}. 

Nevertheless, in realistic settings, physical systems are inevitably coupled to their environments and must be described by mixed states rather than by pure states, as in isolated systems. This makes the study of phases of matter in mixed states a subject of both fundamental and practical significance. Recent work has uncovered a variety of mixed-state quantum phases, including strong-to-weak spontaneous symmetry breaking (SWSSB) phases, average symmetry-protected topological (ASPT) phases, finite-temperature topological phases, and intrinsic topological phases that arise only in open quantum systems~\cite{buvca2012note,de_Groot_2022, MaWangASPT, Ma2025topological, fan2024diagnostics, LeeYouXu2022, ZhangQiBi2022, Lee_2023, Zou_2023, Zhang_2023, ma2024symmetry, xue2024tensor, guo2024locally, guo2024designing, chen2023symmetryenforced, chen2024separability, chen2024unconventional, lessa2024mixedstate, wang2023intrinsic, wang2024anomaly, sala2024spontaneous, Lessa_2025, huang2024hydro, OgunnaikeFeldmeierLee2023, zhang2024quantumcommunicationmixedstateorder, zhang2024fluctuation,
Ellison_2025, Sohal_2025, zhang2024strong, Sala_2025, Sun_2025, Stephen_2025, Su_2024, sang2024mixed, sang2024stability, moharramipour2024symmetry, gu2024spontaneous, guo2024new, luo2025topological, schafernameki2025symtft, qi2025symmetrytaco, aldossari2025tensor, yang2025topological,sun2025anomalous, ma2025circuit, mana2026mixed}. 

Previous studies of mixed-state topological orders have primarily focused on abelian cases. Recently, decoherence of Kitaev's quantum double model attracted considerable attention~\cite{Sohal_2025, Ellison_2025, yang2025topological}. However, a lattice realization of the SWSSB mechanism on these systems remains unidentified. In this work, we investigate the (untwisted) Kitaev’s quantum double model $\mathcal{D}(G)$ with a general finite group $G$, subject to external decoherence. We show that the decohered quantum-double state exhibits an SWSSB of a strong electric $1$-form non-invertible symmetry and weak-to-trivial spontaneous symmetry breaking (WTSSB) of a dual weak magnetic $1$-form non-invertible symmetry. We further prove that the decohered quantum-double density matrices form a locally indistinguishable set. At maximal decoherence, these density matrices further constitute an information convex set. The number of its extremal points equals the ground-state degeneracy of the original quantum-double model, so the quantum information stored in the ground-state subspace is reduced to classical information stored in the convex set. 

\textit{Symmetries of density matrices} -- Before studying the decoherence of general non-abelian quantum double models, we first introduce the types of symmetries that a density matrix can have.

One perspective from which to study phases of matter and their phase transitions is through symmetry and symmetry breaking, an approach known as the Landau paradigm. Topological phases of matter were historically believed to lie outside the Landau paradigm. However, recent studies of generalized symmetries have extended the Landau paradigm to a broader form that incorporates higher-form symmetries as well as non-invertible symmetries. Passing from pure states to density matrices extends the paradigm in a further direction, since a symmetry may now be realized in two inequivalent ways.

For density matrices, there are accordingly two types of symmetries, namely strong symmetry and weak symmetry. A strong invertible symmetry $U_s$ is defined by the condition
\begin{equation}
    \begin{aligned}
        U_s \rho = e^{i \theta} \rho,
    \end{aligned}
\end{equation}
while a weak invertible symmetry $U_w$ is defined by
\begin{equation}
    \begin{aligned}
        U_w \rho U_w^{\dagger} = \rho.
    \end{aligned}
\end{equation}
One can therefore see that a strong symmetry is automatically a weak symmetry, while a weak symmetry is in general not a strong symmetry. Physically, a strong symmetry is a symmetry of the quantum state itself, while a weak symmetry is a symmetry of the ensemble. This distinction is what makes decoherence a symmetry-based question: a channel that preserves a symmetry only at the level of the ensemble reduces a strong symmetry to a weak one, and the resulting mixed state may realize a phase distinct from that of the original pure state.

In later sections, we generalize the strong and weak symmetries to 1-form non-invertible symmetries on decohered non-abelian quantum double models. Before that, however, we first review how the toric code behaves under decoherence and what its symmetry properties are.

\textit{Abelian case: decohering the toric code} -- Consider a square lattice with qubit degrees of freedom assigned to the edges. The Hamiltonian is given by
\begin{equation}
    \begin{aligned}
        H_{TC} = - \sum_{v} A_v - \sum_p B_p,
    \end{aligned}
\end{equation}
where
\begin{equation}
    \begin{aligned}
  A_v =    \begin{tikzpicture}
    [scale= 1, baseline={([yshift=-2.5pt]current bounding box.center)}]
        \draw[step = 2] (-1.2,-1.2) grid (1.2,1.2);
        \node[inner sep=0.5pt, fill=white] at (-0.7, 0) {\footnotesize $X$};
        \node[inner sep=0.5pt, fill=white] at (0.7,0) {\footnotesize $X$};
        \node[inner sep=0.5pt, fill=white] at (0, 0.7) {\footnotesize $X$};
        \node[inner sep=0.5pt, fill=white] at (0,-0.7) {\footnotesize $X$};  
        \node[inner sep=0.5pt, fill=white] at (0.2, -0.2) {\footnotesize $v$};
    \end{tikzpicture}\ ,\quad
    B_p = \begin{tikzpicture}
    [scale= 0.8, baseline={([yshift=-2.5pt]current bounding box.center)}]
    \draw[step = 2] (0,0) grid (2,2);
        \node[inner sep=0.5pt, fill=white] at (0, 1) {\footnotesize $Z$};
        \node[inner sep=0.5pt, fill=white] at (2,1) {\footnotesize $Z$};
        \node[inner sep=0.5pt, fill=white] at (1, 0) {\footnotesize $Z$};
        \node[inner sep=0.5pt, fill=white] at (1,2) {\footnotesize $Z$};
        \node[inner sep=0.5pt, fill=white] at (1,1) {\footnotesize $p$};
    \end{tikzpicture}.
\end{aligned}
\end{equation}
On an infinite plane, the density matrix can be written as follows:
\begin{equation}
    \begin{aligned}
        \rho_{TC} = \prod_v \frac{1+A_v}{2} \prod_p \frac{1+B_p}{2}
    \end{aligned}
\end{equation}
One can verify that this density matrix possesses two kinds of strong 1-form symmetries:
\begin{equation}
    \begin{aligned}
        W_{\gamma}^{Z} \rho_{TC} = \rho_{TC},\quad W_{\gamma^*}^{X} \rho_{TC} = \rho_{TC},
    \end{aligned}
\end{equation}
where
\begin{equation}
    \begin{aligned}
        W_{\gamma}^{Z} = \prod_{e \in \gamma} Z_e, \quad W_{\gamma^*}^{X} = \prod_{e \in \gamma^*} X_e,
    \end{aligned}
\end{equation}
and $\gamma$ and $\gamma^*$ are closed loops on the lattice and the dual lattice, respectively. Physically, these two operators correspond to the electric and the magnetic anyon lines.

For a $1$-form symmetry, the charged order parameter is a line operator rather than a local operator. In the toric code, the electric and magnetic lines are charged under each other's $1$-form symmetries through their nontrivial mutual braiding. At the toric-code fixed point, the expectation values of contractible electric and magnetic loops remain equal to $1$ as the loops become arbitrarily large. This is the higher-form analogue of long-range order and identifies the topological phase as a phase with spontaneously broken $1$-form symmetries~\cite{McGreevy_2023,Gaiotto_2015}.

The same mutual braiding gives rise to a mixed 't Hooft anomaly between the two $1$-form symmetries, as illustrated in Fig.~\ref{Fig:anomaly}. On a torus, this anomaly is reflected in the algebra of the noncontractible loop operators, which act as logical Pauli operators on the two encoded qubits. Electric and magnetic logical loops with an odd number of intersections anticommute, so no pure logical state can be a simultaneous eigenstate of both sets of operators. In particular, a logical state invariant under the noncontractible $X$ loops along both directions is not invariant under the corresponding $Z$ loops, and vice versa.

\begin{figure}
\begin{tikzpicture}[use Hobby shortcut,xscale=1.5,yscale=1.5]
\tikzstyle{sergio}=[rectangle,draw=none]
\coordinate[] (A6) at (1,0.7);
\coordinate[] (A7) at (2,0.6);
\coordinate[] (A8) at (2.5,0.8);
\coordinate[] (A9) at (3,1.2);
\coordinate[] (A10) at (3.8,1.3);
\draw[color=green!40,line width=7pt] (A6) .. (A7) .. (A8) .. (A9) .. (A10);
\path (4.1,1.3) node [style=sergio]{\large$s_1$};
\path (0.7,0.7) node [style=sergio]{\large$s_0$};
\draw[red!30, line width=7pt, dash pattern={on 6pt off 3pt}] (3.8,1.3)circle (1);
\path (2,1.1) node [style=sergio]{\large$W_{\gamma}^Z$};
\path (2.8, 0.3) node [style=sergio]{\large$W_{\gamma*}^X$};
\path (5.5, 1) node [style=sergio]{\large$=-1$};
\end{tikzpicture}
\caption{Illustration of 't Hooft anomaly/SSB of 1-form ribbon symmetries. The open green ribbon depicts $W_{\gamma}^Z$, the closed dashed red ribbon depicts $W_{\gamma*}^X$. When these two anyon line has odd number of intersection, a $-1$ phase appear.}
\label{Fig:anomaly}
\end{figure}

In the literature, three types of decoherence of the toric code are commonly considered: the Pauli $X$ channel, the Pauli $Z$ channel, and the nearest-neighbor $XZ$ channel. Physically, these three channels correspond respectively to the incoherent proliferation of the magnetic anyon, the electric anyon, and the fermionic anyon~\cite{Ellison_2025}. For concreteness, we take the Pauli $Z$ channel as our example.

On every edge $e$, we define the channel
\begin{equation}
    \begin{aligned}
        \mathcal{N}_e[\rho_{TC}] = (1-p) \rho_{TC} + p Z_e \rho_{TC} Z_e,
    \end{aligned}
\end{equation}
and the full channel is $\mathcal{N} = \bigotimes_e \mathcal{N}_e$. In the maximal decoherence limit $p = \frac{1}{2}$, we have
\begin{equation}
    \begin{aligned}
        \rho^{\mathrm{dec}}_{TC} = \mathcal{N}[\rho_{TC}]|_{p = \frac{1}{2}} = \prod_p \frac{1 + B_p}{2}.
    \end{aligned}
\end{equation}
To save notation, in the remainder of this paper we use $\mathcal{N}$ to denote the maximally decohering channel unless otherwise specified.

The decohered density matrix $\rho^{\mathrm{dec}}_{TC}$ remains symmetric under the strong electric 1-form symmetry
\begin{equation}
    \begin{aligned}
        W_{\gamma}^Z \rho^{\mathrm{dec}}_{TC} = \rho^{\mathrm{dec}}_{TC}.
    \end{aligned}
\end{equation}
However, the original strong magnetic 1-form symmetry becomes weak.
\begin{equation}
    \begin{aligned}
        W_{\gamma*}^X \rho^{\mathrm{dec}}_{TC} W_{\gamma*}^X = \rho^{\mathrm{dec}}_{TC}. \label{eq:TC_weak_symmetry}
    \end{aligned}
\end{equation}
To make the corresponding long-range order explicit, let $\gamma_R^*$ be a closed, contractible dual loop of linear size $R$. Eq.~\eqref{eq:TC_weak_symmetry} gives
\begin{equation}
    \begin{aligned}
        \mathcal{F}\!\left( \rho_{\mathrm{TC}}^{\mathrm{dec}}, W_{\gamma_R^*}^{X}\rho_{\mathrm{TC}}^{\mathrm{dec}} W_{\gamma_R^*}^{X} \right)=1
    \end{aligned}
\end{equation}
for every $R$, whereas
\begin{equation}
    \begin{aligned}
        \operatorname{Tr}\!\left( \rho_{\mathrm{TC}}^{\mathrm{dec}}W_{\gamma_R^*}^{X} \right)=0.
    \end{aligned}
\end{equation}
Thus, the magnetic loop exhibits long-range fidelity order in the absence of conventional loop order, and, together with the mixed anomaly illustrated in Fig.~\ref{Fig:anomaly}, provides the characteristic diagnostic of $1$-form SWSSB.

Microscopically, maximal $Z$ decoherence removes all off-diagonal components of a density matrix in the computational basis,
\begin{equation}
    \begin{aligned}
        \mathcal{N}[|g\rangle \langle h |] = \delta_{g,h} |g\rangle \langle h |,\quad g, h \in\{0,1\}.
    \end{aligned}
\end{equation}
Now consider the toric code defined on a topologically nontrivial manifold, where a general logical pure state is given by
\begin{equation}
    \begin{aligned}
        |\psi\rangle = \sum_a c_a |\psi_a\rangle,
    \end{aligned}
\end{equation}
where $|\psi_a\rangle$ denotes a basis state of the degenerate ground-state subspace. After maximal $Z$ decoherence, we obtain
\begin{equation}
    \begin{aligned}
        \mathcal{N}[|\psi\rangle\langle\psi|] = \sum_a |c_a|^2 \rho_a,
    \end{aligned}
\end{equation}
where
\begin{equation}
    \begin{aligned}
        \rho_a = \frac{1}{\mathcal{Z}_a}\sum_{\gamma^*}W_{\gamma_a^*}^X|\gamma^*\rangle \langle \gamma^* |W_{\gamma_a^*}^X,
    \end{aligned}
\end{equation}
with $\gamma^*$ denoting contractible loop configurations and $\gamma_a^*$ a noncontractible loop configuration in the homology class labeled by $a$. Here $|\gamma^*\rangle$ is the state in which the qubits on $\gamma^*$ are in $|1\rangle$ and all remaining qubits are in $|0\rangle$. Therefore, $\rho_a$ is the ensemble of all classical loop configurations within the homology class labeled by $a$, while $\mathcal{N}[|\psi\rangle\langle\psi|]$ is the ensemble of all classical loop configurations, with the probability of each homology class $a$ given by $|c_a|^2$.

Mathematically, as we will show in Theorems~\ref{thm_1} and~\ref{thm_2} and in Corollary~\ref{corollary_3}, under maximal charge decoherence the original ground-state subspace degrades to an information convex set, in which each original pure state of a fixed homology class becomes an extremal point. We defer the general theorems until after we introduce the decoherence of the general finite group quantum double model.

\textit{A brief review of quantum double $\mathcal{D}(G)$} -- The quantum double $\mathcal{D}(G)$ can be defined on an oriented square lattice $\Lambda=(V,E,F)$ by associating each edge with the Hilbert space $\mathbb{C}[G]$ spanned by $\{|g\rangle | g \in G \}$. The commuting projector Hamiltonian is given as follows, 
\begin{align}
    H=-\sum_{v\in V} A_v -\sum_{p\in F} B_p,
    \label{Eq: Hamiltonian}
\end{align}
where $A_v=\frac{1}{|G|}\sum_{g \in G} A_v^g$. The actions of the Hamiltonian terms $A_v^g$ and $B_p$ are defined as follows, 
\begin{equation}
\begin{tikzpicture}[scale=1.15]
\tikzstyle{sergio}=[rectangle,draw=none]

\draw (-1,0)--(0,0) node[currarrow,pos=0.5,sloped]{};
\draw (0,0)--(1,0) node[currarrow,pos=0.5,sloped]{};
\draw (0,0)--(0,-1) node[currarrow,pos=0.5,sloped]{};
\draw (0,1)--(0,0) node[currarrow,pos=0.5,sloped]{};

\path (-1.4,0) node [style=sergio]{$A_v^g$};
\path (-0.5,0.2) node [style=sergio]{$x_1$};
\path (0.5,0.2) node [style=sergio]{$x_2$};
\path (-0.3,0.6) node [style=sergio]{$y_1$};
\path (-0.3,-0.4) node [style=sergio]{$y_2$};
\path (1.5,0) node [style=sergio]{$=$};

\draw (2,0)--(3,0) node[currarrow,pos=0.5,sloped]{};
\draw (3,0)--(4,0) node[currarrow,pos=0.5,sloped]{};
\draw (3,0)--(3,-1) node[currarrow,pos=0.5,sloped]{};
\draw (3,1)--(3,0) node[currarrow,pos=0.5,sloped]{};

\path (2.5,0.2) node [style=sergio]{$x_1 \bar{g}$};
\path (3.5,0.2) node [style=sergio]{$g x_2$};
\path (3-0.3,0.6) node [style=sergio]{$y_1 \bar{g}$};
\path (3-0.3,-0.4) node [style=sergio]{$g y_2$};

\path (0.1,0.15) node [style=sergio]{$v$};
\path (3.1, 0.15) node [style=sergio]{$v$};

\draw (-0.5-0.2,0.5-2)--(0.5-0.2,0.5-2) node[currarrow,pos=0.5,sloped]{};
\draw (0.5-0.2,0.5-2)--(0.5-0.2,-0.5-2) node[currarrow,pos=0.5,sloped]{};
\draw (-0.5-0.2,0.5-2)--(-0.5-0.2,-0.5-2) node[currarrow,pos=0.5,sloped]{};
\draw (-0.5-0.2,-0.5-2)--(0.5-0.2,-0.5-2) node[currarrow,pos=0.5,sloped]{};

\path (-1.2-0.2,0-2) node [style=sergio]{$B_p$};
\path (-0.7-0.2,0-2) node [style=sergio]{$y_1$};
\path (0.7-0.2,0-2) node [style=sergio]{$y_2$};
\path (0-0.2,0.7-2) node [style=sergio]{$x_1$};
\path (0-0.2,-0.7-2) node [style=sergio]{$x_2$};
\path (1.5,0-2) node [style=sergio]{\small$= \delta_{x_1 y_2 \bar{x}_2 \bar{y}_1, e}$};

\draw (-0.5+3.2,0.5-2)--(0.5+3.2,0.5-2) node[currarrow,pos=0.5,sloped]{};
\draw (0.5+3.2,0.5-2)--(0.5+3.2,-0.5-2) node[currarrow,pos=0.5,sloped]{};
\draw (-0.5+3.2,0.5-2)--(-0.5+3.2,-0.5-2) node[currarrow,pos=0.5,sloped]{};
\draw (-0.5+3.2,-0.5-2)--(0.5+3.2,-0.5-2) node[currarrow,pos=0.5,sloped]{};

\path (-0.7+3.2,0-2) node [style=sergio]{$y_1$};
\path (0.7+3.2,0-2) node [style=sergio]{$y_2$};
\path (0+3.2,0.7-2) node [style=sergio]{$x_1$};
\path (0+3.2,-0.7-2) node [style=sergio]{$x_2$};

\path (0.1-0.3,0.15-0.5-2) node [style=sergio]{$v$};
\path (3.1-0.3, 0.15-0.5-2) node [style=sergio]{$v$};

\end{tikzpicture},\label{Eq: AvBp}
\end{equation}
where $\bar{g} = g^{-1}$. The ground state satisfies $\ket{\psi_{\mathrm{QD}}}=A_v\ket{\psi_{\mathrm{QD}}}=B_p\ket{\psi_{\mathrm{QD}}}$. In particular, in the group basis $\left\{\ket{g}\big|g\in G\right\}$, we could reformulate the vertex and plaquette terms $A_v$ and $B_p$ in terms of generalized Pauli $X$ and $Z$ operators, namely,
\begin{align}
\begin{gathered}
L_h^{+}=\sum_{g\in G}\ket{hg}\bra{g},~L_h^{-}=\sum_{g\in G}\ket{g\bar{h}}\bra{g},\\
Z_\Gamma=\sum_{g\in G}\Gamma(g)\otimes\ket{g}\bra{g}.
\end{gathered}
\label{Eq: X and Z}
\end{align}
Therefore, we have
\begin{align}
L_h^{+}\ket{g}=\ket{hg},~L_h^{-}\ket{g}=\ket{g\bar{h}},~Z_\Gamma\ket{g}=\Gamma(g)\ket{g},
\end{align}
where $\Gamma\in \mathrm{Rep}(G)$ is an irreducible representation (irrep) of the group $G$.

The vertex term $A_v^g$ [cf. Eq. \eqref{Eq: AvBp}] can thus be reformulated as follows,
\begin{align}
A_v^g=\prod_{e\ni v}L_{g}^{\pm},
\label{Eq: AvX}
\end{align}
where $\pm$ is determined by the orientation of the lattice. Similarly, because of the great orthogonality theorem \cite{supp}, the plaquette term $B_p$ [cf. Eq. \eqref{Eq: AvBp}] can be reformulated in terms of $Z_\Gamma$,
\begin{align}
B_p=\frac{1}{|G|}\sum_{\Gamma\in\mathrm{Rep}(G)}d_\Gamma\cdot\tr\left(\prod_{e\in p}Z_{\Gamma,e}^{\pm}\right),
\label{Eq: BpZ}
\end{align}
where $d_\Gamma$ is the dimension of the irrep $\Gamma$, $Z_\Gamma^+=Z_\Gamma$, $Z_\Gamma^{-}=Z_\Gamma^\dag$, and $\pm$ is determined by the orientation of the lattice. $\tr \left(\prod_{e\in p}Z_{\Gamma,e}^{\pm}\right)$ is a matrix product operator (MPO)~\cite{fannes1992finitely,pirvu2010matrix}, where the trace is taken over the virtual space.

\begin{figure}
\begin{tikzpicture}[scale=1.2]
\tikzstyle{sergio}=[rectangle,draw=none]
\filldraw[fill=gray!50, draw=gray, dotted] (-3.5,-4.5) -- (-3,-4) -- (0,-4) -- (0.5,-4.5) -- cycle;
\draw[draw=gray, dotted] (-3,-4) -- (-2.5,-4.5) -- (-2,-4) -- (-1.5,-4.5) -- (-1,-4) -- (-0.5,-4.5) -- (0,-4);
\draw (-4,-4)--(-3,-4) node[currarrow, pos=0.5, sloped] {};
\draw (-3,-4)--(-2,-4) node[currarrow, pos=0.5, sloped] {};
\draw (-2,-4)--(-1,-4) node[currarrow, pos=0.5, sloped] {};
\draw (-1,-4)--(0,-4) node[currarrow, pos=0.5, sloped] {};
\draw (0,-4)--(1,-4) node[currarrow, pos=0.5, sloped] {};
\draw (-4,-4)--(-4,-5) node[currarrow, pos=0.5, sloped] {};
\draw (-3,-4)--(-3,-5) node[currarrow, pos=0.5, sloped] {};
\draw (-2,-4)--(-2,-5) node[currarrow, pos=0.5, sloped] {};
\draw (-1,-4)--(-1,-5) node[currarrow, pos=0.5, sloped] {};
\draw (0,-4)--(0,-5) node[currarrow, pos=0.5, sloped] {};
\draw (1,-4)--(1,-5) node[currarrow, pos=0.5, sloped] {};
\draw (-4,-5)--(-3,-5) node[currarrow, pos=0.5, sloped] {};
\draw (-3,-5)--(-2,-5) node[currarrow, pos=0.5, sloped] {};
\draw (-2,-5)--(-1,-5) node[currarrow, pos=0.5, sloped] {};
\draw (-1,-5)--(0,-5) node[currarrow, pos=0.5, sloped] {};
\draw (0,-5)--(1,-5) node[currarrow, pos=0.5, sloped] {};
\path (-3.5,-4.25) node [style=sergio]{$s_0$};
\path (-2.5,-3.8) node [style=sergio]{$x_1$};
\path (-1.5,-3.8) node [style=sergio]{$x_2$};
\path (-0.5,-3.8) node [style=sergio]{$x_3$};
\path (-3,-4.8) node [style=sergio]{\small $hy_1$};
\path (-2,-4.8) node [style=sergio]{\small $\bar{x}_1'h{x}_1'y_2$};
\path (-0.8,-4.8) node [style=sergio]{\small $\bar{x}_2'h{x}_2'y_3$};
\path (0.4,-4.8) node [style=sergio]{\small $\bar{x}_3'h{x}_3'y_4$};
\path (0.5,-4.25) node [style=sergio]{$s_1$};
\path (-5,-4.5) node [style=sergio]{\large $=\delta_{g,x_1x_2x_3}$};
\filldraw[fill=gray!50, draw=gray, dotted] (-3.5,-4.5+1.5) -- (-3,-4+1.5) -- (0,-4+1.5) -- (0.5,-4.5+1.5) -- cycle;
\draw[draw=gray, dotted] (-3,-4+1.5) -- (-2.5,-4.5+1.5) -- (-2,-4+1.5) -- (-1.5,-4.5+1.5) -- (-1,-4+1.5) -- (-0.5,-4.5+1.5) -- (0,-4+1.5);
\draw (-4,-4+1.5)--(-3,-4+1.5) node[currarrow, pos=0.5, sloped] {};
\draw (-3,-4+1.5)--(-2,-4+1.5) node[currarrow, pos=0.5, sloped] {};
\draw (-2,-4+1.5)--(-1,-4+1.5) node[currarrow, pos=0.5, sloped] {};
\draw (-1,-4+1.5)--(0,-4+1.5) node[currarrow, pos=0.5, sloped] {};
\draw (0,-4+1.5)--(1,-4+1.5) node[currarrow, pos=0.5, sloped] {};
\draw (-4,-4+1.5)--(-4,-5+1.5) node[currarrow, pos=0.5, sloped] {};
\draw (-3,-4+1.5)--(-3,-5+1.5) node[currarrow, pos=0.5, sloped] {};
\draw (-2,-4+1.5)--(-2,-5+1.5) node[currarrow, pos=0.5, sloped] {};
\draw (-1,-4+1.5)--(-1,-5+1.5) node[currarrow, pos=0.5, sloped] {};
\draw (0,-4+1.5)--(0,-5+1.5) node[currarrow, pos=0.5, sloped] {};
\draw (1,-4+1.5)--(1,-5+1.5) node[currarrow, pos=0.5, sloped] {};
\draw (-4,-5+1.5)--(-3,-5+1.5) node[currarrow, pos=0.5, sloped] {};
\draw (-3,-5+1.5)--(-2,-5+1.5) node[currarrow, pos=0.5, sloped] {};
\draw (-2,-5+1.5)--(-1,-5+1.5) node[currarrow, pos=0.5, sloped] {};
\draw (-1,-5+1.5)--(0,-5+1.5) node[currarrow, pos=0.5, sloped] {};
\draw (0,-5+1.5)--(1,-5+1.5) node[currarrow, pos=0.5, sloped] {};
\path (-3.5,-4.25+1.5) node [style=sergio]{$s_0$};
\path (-2.5,-3.8+1.5) node [style=sergio]{$x_1$};
\path (-1.5,-3.8+1.5) node [style=sergio]{$x_2$};
\path (-0.5,-3.8+1.5) node [style=sergio]{$x_3$};
\path (0.5,-2.75) node [style=sergio]{$s_1$};
\path (-3,-2.8-0.5) node [style=sergio]{$y_1$};
\path (-2,-2.8-0.5) node [style=sergio]{$y_2$};
\path (-1,-2.8-0.5) node [style=sergio]{$y_3$};
\path (-0,-2.8-0.5) node [style=sergio]{$y_4$};
\path (-4.7,-3) node [style=sergio]{\large $F_{\xi,\Delta}^{h,g}$};
\filldraw[fill=gray!50, draw=gray, dotted] (-3.5,-4.5+3.5) -- (-3,-5+3.5) -- (0,-5+3.5) -- (0.5,-4.5+3.5) -- cycle;
\draw[draw=gray, dotted] (-3,-4+3.5) -- (-2.5,-4.5+3.5) -- (-2,-4+3.5) -- (-1.5,-4.5+3.5) -- (-1,-4+3.5) -- (-0.5,-4.5+3.5) -- (0,-4+3.5);
\draw (-4,-4+3.5)--(-3,-4+3.5) node[currarrow, pos=0.5, sloped, rotate=180] {};
\draw (-3,-4+3.5)--(-2,-4+3.5) node[currarrow, pos=0.5, sloped, rotate=180] {};
\draw (-2,-4+3.5)--(-1,-4+3.5) node[currarrow, pos=0.5, sloped, rotate=180] {};
\draw (-1,-4+3.5)--(0,-4+3.5) node[currarrow, pos=0.5, sloped, rotate=180] {};
\draw (0,-4+3.5)--(1,-4+3.5) node[currarrow, pos=0.5, sloped, rotate=180] {};
\draw (-4,-5+3.5)--(-4,-4+3.5) node[currarrow, pos=0.5, sloped] {};
\draw (-3,-5+3.5)--(-3,-4+3.5) node[currarrow, pos=0.5, sloped] {};
\draw (-2,-5+3.5)--(-2,-4+3.5) node[currarrow, pos=0.5, sloped] {};
\draw (-1,-5+3.5)--(-1,-4+3.5) node[currarrow, pos=0.5, sloped] {};
\draw (0,-5+3.5)--(0,-4+3.5) node[currarrow, pos=0.5, sloped] {};
\draw (1,-5+3.5)--(1,-4+3.5) node[currarrow, pos=0.5, sloped] {};
\draw (-4,-5+3.5)--(-3,-5+3.5) node[currarrow, pos=0.5, sloped, rotate=180] {};
\draw (-3,-5+3.5)--(-2,-5+3.5) node[currarrow, pos=0.5, sloped, rotate=180] {};
\draw (-2,-5+3.5)--(-1,-5+3.5) node[currarrow, pos=0.5, sloped, rotate=180] {};
\draw (-1,-5+3.5)--(0,-5+3.5) node[currarrow, pos=0.5, sloped, rotate=180] {};
\draw (0,-5+3.5)--(1,-5+3.5) node[currarrow, pos=0.5, sloped, rotate=180] {};
\path (-3.5,-4.25+3) node [style=sergio]{$s_0$};
\path (-2.5,-3.8+2.1) node [style=sergio]{$x_1$};
\path (-1.5,-3.8+2.1) node [style=sergio]{$x_2$};
\path (-0.5,-3.8+2.1) node [style=sergio]{$x_3$};
\path (-3,-4.8+4) node [style=sergio]{\small $y_1h$};
\path (-2,-4.8+4) node [style=sergio]{\small $y_2\bar{x}_1'h{x}_1'$};
\path (-0.8,-4.8+4) node [style=sergio]{\small $y_3\bar{x}_2'h{x}_2'$};
\path (0.4,-4.8+4) node [style=sergio]{\small $y_4\bar{x}_3'h{x}_3'$};
\path (0.5,-4.25+3) node [style=sergio]{$s_1$};
\path (-5,-4.5+3.5) node [style=sergio]{\large $=\delta_{g,\overline{x_1 x_2 x_3}}$};
\filldraw[fill=gray!50, draw=gray, dotted] (-3.5,-4.5+1.5+3.5) -- (-3,-5+1.5+3.5) -- (0,-5+1.5+3.5) -- (0.5,-4.5+1.5+3.5) -- cycle;
\draw[draw=gray, dotted] (-3,-4+1.5+3.5) -- (-2.5,-4.5+1.5+3.5) -- (-2,-4+1.5+3.5) -- (-1.5,-4.5+1.5+3.5) -- (-1,-4+1.5+3.5) -- (-0.5,-4.5+1.5+3.5) -- (0,-4+1.5+3.5);
\draw (-4,-4+1.5+3.5)--(-3,-4+1.5+3.5) node[currarrow, pos=0.5, sloped, rotate=180] {};
\draw (-3,-4+1.5+3.5)--(-2,-4+1.5+3.5) node[currarrow, pos=0.5, sloped, rotate=180] {};
\draw (-2,-4+1.5+3.5)--(-1,-4+1.5+3.5) node[currarrow, pos=0.5, sloped, rotate=180] {};
\draw (-1,-4+1.5+3.5)--(0,-4+1.5+3.5) node[currarrow, pos=0.5, sloped, rotate=180] {};
\draw (0,-4+1.5+3.5)--(1,-4+1.5+3.5) node[currarrow, pos=0.5, sloped, rotate=180] {};
\draw (-4,-5+1.5+3.5)--(-4,-4+1.5+3.5) node[currarrow, pos=0.5, sloped] {};
\draw (-3,-5+1.5+3.5)--(-3,-4+1.5+3.5) node[currarrow, pos=0.5, sloped] {};
\draw (-2,-5+1.5+3.5)--(-2,-4+1.5+3.5) node[currarrow, pos=0.5, sloped] {};
\draw (-1,-5+1.5+3.5)--(-1,-4+1.5+3.5) node[currarrow, pos=0.5, sloped] {};
\draw (0,-5+1.5+3.5)--(0,-4+1.5+3.5) node[currarrow, pos=0.5, sloped] {};
\draw (1,-5+1.5+3.5)--(1,-4+1.5+3.5) node[currarrow, pos=0.5, sloped] {};
\draw (-4,-5+1.5+3.5)--(-3,-5+1.5+3.5) node[currarrow, pos=0.5, sloped, rotate=180] {};
\draw (-3,-5+1.5+3.5)--(-2,-5+1.5+3.5) node[currarrow, pos=0.5, sloped, rotate=180] {};
\draw (-2,-5+1.5+3.5)--(-1,-5+1.5+3.5) node[currarrow, pos=0.5, sloped, rotate=180] {};
\draw (-1,-5+1.5+3.5)--(0,-5+1.5+3.5) node[currarrow, pos=0.5, sloped, rotate=180] {};
\draw (0,-5+1.5+3.5)--(1,-5+1.5+3.5) node[currarrow, pos=0.5, sloped, rotate=180] {};
\path (-3.5,-4.25+1+3.5) node [style=sergio]{$s_0$};
\path (-2.5,-3.8+1.5+2.1) node [style=sergio]{$x_1$};
\path (-1.5,-3.8+1.5+2.1) node [style=sergio]{$x_2$};
\path (-0.5,-3.8+1.5+2.1) node [style=sergio]{$x_3$};
\path (0.5,-2.25+2.5) node [style=sergio]{$s_1$};
\path (-3,-2.8+3.5) node [style=sergio]{$y_1$};
\path (-2,-2.8+3.5) node [style=sergio]{$y_2$};
\path (-1,-2.8+3.5) node [style=sergio]{$y_3$};
\path (-0,-2.8+3.5) node [style=sergio]{$y_4$};
\path (-4.7,-2.5+3) node [style=sergio]{\large $F_{\xi,\nabla}^{h,g}$};

\draw[draw=gray, dotted] (-4,-5.5) -- (-3.5,-6) -- (-3,-5.5);
\filldraw[fill=gray!50, draw=gray, dotted] (-4,-5.5) -- (-3.5,-6) -- (-3,-5.5);
\path (-3.5,-5.7) node [style=sergio]{ $\tau$};

\draw (-4,-5.5)--(-3,-5.5);
\draw (-4,-6.5)--(-3,-6.5);
\draw (-4,-5.5)--(-4,-6.5);
\draw (-3,-5.5)--(-3,-6.5);

\draw[draw=gray, dotted] (-2,-5.5) -- (-1.5,-6) -- (-2,-5.5);
\draw[draw=gray, dotted] (2-2,-5.5) -- (1-1.5,-6) -- (1-2,-5.5);
\filldraw[fill=gray!50, draw=gray, dotted] (0.5-2,-6) -- (0.5-1.5,-5.5) -- (0.5-1,-6);
\path (-1.2,-5.8) node [style=sergio]{ $\tau*$};

\draw (2-4,-5.5)--(2-3,-5.5);
\draw (2-4,-6.5)--(2-3,-6.5);
\draw (2-4,-5.5)--(2-4,-6.5);
\draw (2-3,-5.5)--(2-3,-6.5);

\draw (3-4,-5.5)--(3-3,-5.5);
\draw (3-4,-6.5)--(3-3,-6.5);
\draw (3-3,-5.5)--(3-3,-6.5);

\end{tikzpicture}
    \caption{The action of ribbon operators $F_{\xi,\Delta}^{(h,g)}$ and $F_{\xi,\nabla}^{(h,g)}$ starts from $s_0$ to $s_1$. Here we have $x_i'=x_1\cdots x_i$. On the bottom, we show examples of a direct triangle (left) and a dual triangle (right). Each triangle is associated with one qudit living on the corresponding edge.}
    \label{fig:Ribbon_1}
\end{figure}
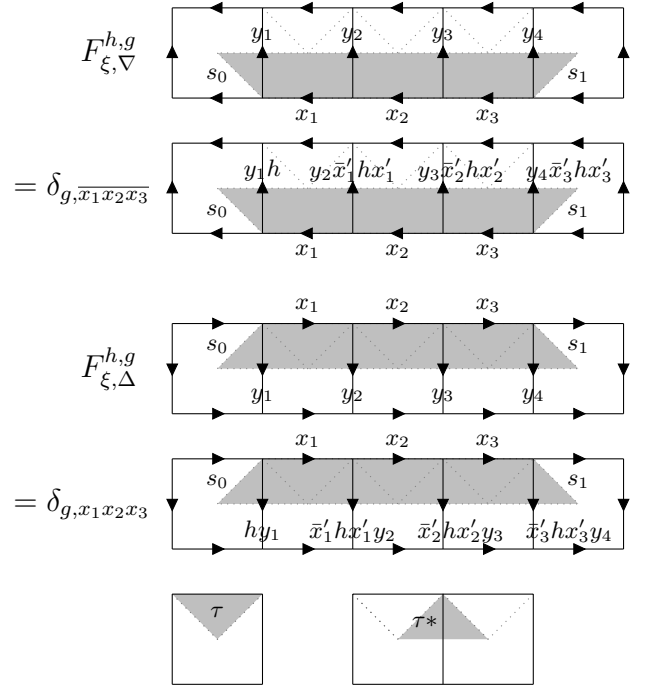

The definition of ribbon operators is illustrated in Fig.~\ref{fig:Ribbon_1}. In particular, a ribbon $\xi$ consists of a sequence of direct and dual triangles along the path~\cite{Kitaev_2003,Bombin2007,beigi2010quantum-2ab}. For ribbon operators of the form $F_\xi^{e,g}$, the action on every dual triangle $\tau^* \in \xi$ is trivial, allowing us to reformulate the ribbon operators in terms of generalized $Z$ operators, namely
\begin{align}
F_\xi^{e,g}=\frac{1}{|G|}\sum_{\Gamma\in\mathrm{Rep}(G)}d_\Gamma\cdot\tr\left(\Gamma(\bar{g})\prod_{\tau\in\xi}Z_{\Gamma,\tau}^{\pm}\right),
\label{Eq: magnetic ribbon}
\end{align}
where the product is taken over all qudits on the direct triangles along the path. The action of the ribbon operator $F_\xi^{e,g}$ yields a Kronecker delta $\delta \left(g, \prod_{\tau \in \xi} g_{e_\tau}\right)$ by the great orthogonality theorem, consistent with the definition in Fig.~\ref{fig:Ribbon_1}. The Fourier transform of the ribbon operators gives the
electric anyonic ribbon operator
\begin{equation}
    \begin{aligned}
    F_{\xi}^\Gamma := \frac{d_\Gamma}{|G|} \sum_{g\in G} \chi_\Gamma(g)F_\xi^{e,g}.
    \end{aligned}
\end{equation}
This is the conventional normalization of the electric anyon ribbon.
For the discussion of non-invertible symmetry, it is more convenient
to introduce the quantum-dimension-normalized MPO
\begin{equation}
    \begin{aligned}
        \mathcal{W}_\xi^\Gamma := \frac{|G|}{d_\Gamma^2} F_{\xi}^\Gamma = \frac{1}{d_{\Gamma}}\tr \left( \prod_{\tau\in\xi} Z_{\Gamma,\tau}^{\pm} \right) = \frac{\chi_{\Gamma}(g_{\xi})}{d_{\Gamma}}, \label{Eq: strong sym}
    \end{aligned}
\end{equation}
where the signs are determined by the orientation of the ribbon, $g_{\xi} = \prod_e g_e^{\pm}$ is the holonomy along $\xi$, $\chi_{\Gamma}(g_{\xi})$ is the character of $g_{\xi}$, and the trace is taken over the virtual space. Since
\begin{equation}
    \begin{aligned}
        \chi_\Gamma(g)\chi_\Lambda(g) = \sum_\Pi N_{\Gamma\Lambda}^{\Pi} \chi_\Pi(g),
    \end{aligned}
\end{equation}
where $\Gamma, \Lambda, \Pi \in \mathrm{Rep}(G)$, the symmetry MPOs satisfy
\begin{equation}
    \begin{aligned}
        \mathcal{W}_\xi^\Gamma \mathcal{W}_\xi^\Lambda = \sum_\Pi \frac{N_{\Gamma\Lambda}^{\Pi} d_{\Pi}}{d_{\Gamma} d_{\Lambda}} \mathcal{W}_\xi^\Pi, \qquad \mathcal{W}_\xi^{\mathbf{1}} = \mathbbm{1}.
    \end{aligned}
\end{equation}
Thus, the collection $\{\mathcal{W}_{\xi}^{\Gamma}\}_{\Gamma\in\operatorname{Rep}(G)}$ realizes the quantum-dimension-normalized form of the fusion algebra of $\operatorname{Rep}(G)$, with structure constants $N_{\Gamma\Lambda}^{\Pi}d_{\Pi}/(d_{\Gamma}d_{\Lambda})$.

For a contractible closed ribbon $\xi_C$, the flatness condition gives $g_{\xi_C}=e$. Therefore, on any quantum-double ground state $|\psi_{\mathrm{QD}}\rangle$,
\begin{equation}
\mathcal{W}_{\xi_C}^{\Gamma} |\psi_{\mathrm{QD}}\rangle = |\psi_{\mathrm{QD}}\rangle.
\end{equation}
Equivalently, for $\rho_{\mathrm{QD}} = |\psi_{\mathrm{QD}}\rangle \langle\psi_{\mathrm{QD}}|$,
\begin{equation}
\mathcal{W}_{\xi_C}^{\Gamma} \rho_{\mathrm{QD}} = \rho_{\mathrm{QD}}.
\end{equation}
More generally, if
$\rho_{[g_{\xi_C}]}$ is supported in a definite holonomy sector along $\xi_C$, then
\begin{equation}
\mathcal{W}_{\xi_C}^{\Gamma} \rho_{[g_{\xi_C}]} =  \frac{\chi_{\Gamma}(g_{\xi_C})}{d_{\Gamma}}\rho_{[g_{\xi_C}]}.
\end{equation}
Thus, the collection $\{\mathcal{W}_{\xi_C}^{\Gamma}\}_{\Gamma\in\operatorname{Rep}(G)}$ defines a strong $\operatorname{Rep}(G)$-type $1$-form symmetry. For each fixed holonomy sector $[g_{\xi_C}]$, the eigenvalues $\{\chi_{\Gamma}(g_{\xi_C})/d_{\Gamma}\}$ form a one-dimensional representation of the normalized fusion
algebra.

In addition to the
$\operatorname{Rep}(G)$-type $1$-form symmetry operators, one can define an underlying $G$-valued $1$-form symmetry action. Specifically, for a closed ribbon $\xi_C$, define
\begin{equation}
\mathcal{V}_{\xi_C}^{g} := \sum_{k\in G} F_{\xi_C}^{g,k}.
\end{equation}
The ribbon-operator algebra gives
\begin{equation}
\mathcal{V}_{\xi_C}^{g} \mathcal{V}_{\xi_C}^{h} = \mathcal{V}_{\xi_C}^{gh}, \qquad \left( \mathcal{V}_{\xi_C}^{g} \right)^{\dagger} = \mathcal{V}_{\xi_C}^{g^{-1}}, \qquad \mathcal{V}_{\xi_C}^{e} = \mathbbm{1}.
\end{equation}
When $\xi_C$ is contractable, one can verify that
\begin{equation}
\mathcal{V}_{\xi_C}^{g} |\psi_{\mathrm{QD}}\rangle = |\psi_{\mathrm{QD}}\rangle. \label{eq:G_Sym}
\end{equation}
We therefore refer to
$\{\mathcal{V}_{\xi_C}^{g}\}_{g\in G}$ as the underlying $G$-valued $1$-form symmetry operators of the quantum-double model.

For a non-abelian group $G$, an individual operator $\mathcal{V}_{\xi_C}^{g}$ is generally not gauge invariant, since a gauge transformation conjugates its group label $g$. Therefore, a gauge-invariant magnetic anyon line must be constructed by combining group elements belonging to the same conjugacy class.

Let $C$ be a conjugacy class, choose a representative $r_C\in C$, and define its centralizer by
\begin{equation}
Z_C = \{k\in G\mid kr_C=r_Ck\}.
\end{equation}
Choose coset representatives $p_i$ such that $c_i=p_i r_Cp_i^{-1}$. The matrix-valued ribbon operator for the purely magnetic anyon $(C,\mathbf{1})$ is
\begin{equation}
\left( F_{\xi}^{C} \right)_{i,i'} := \frac{1}{|Z_C|} \sum_{k\in Z_C} F_{\xi}^{c_i^{-1},\,p_i k p_{i'}^{-1}}. \label{eq:magnetic_MPO}
\end{equation}
Eq.~\eqref{eq:magnetic_MPO} defines $\left( F_{\xi}^{C} \right)_{i,i'}$ to be an MPO. The corresponding closed magnetic anyon line is obtained by tracing over the virtual space,
\begin{equation}
\begin{aligned}
F_{\xi}^{C} &:= \tr \left[ \left( F_{\xi}^{C} \right)_{i,i'} \right] \\
&= \frac{1}{|Z_C|} \sum_i \sum_{k\in Z_C} F_{\xi}^{c_i^{-1},\,p_i k p_i^{-1}}. \label{Eq:weak_sym}
\end{aligned}
\end{equation}
This operator represents the magnetic anyon line whose flux is labeled by the conjugacy class $C$ and whose centralizer
charge is trivial.

Similar to the electric anyon ribbon, it is useful to distinguish the anyon normalization from the quantum-dimension normalization. We define the quantum-dimension-normalized magnetic ribbon operator by
\begin{equation}
\begin{aligned}
\mathcal{W}_{\xi}^{C} &:= \frac{|Z_C|}{|C|}F_{\xi}^{C}=\frac{|G|}{|C|^2}F_{\xi}^{C}\\
&=\frac{1}{|C|}\sum_i\sum_{k\in Z_C}F_{\xi}^{c_i^{-1},\,p_i k p_i^{-1}}. \label{eq:magnetic_symmetry}
\end{aligned}
\end{equation}
The purely magnetic anyon $(C,\mathbf{1})$ has quantum dimension $d_C=|C|$. Therefore, for a contractible closed magnetic loop, the quantum-dimension-normalized operator acts trivially on the quantum-double ground-state subspace,
\begin{equation}
\begin{aligned}
\mathcal{W}_{\xi_C}^{C}|\psi_{\mathrm{QD}}\rangle=|\psi_{\mathrm{QD}}\rangle.
\end{aligned}
\end{equation}
Equivalently, for $\rho_{\mathrm{QD}}=|\psi_{\mathrm{QD}}\rangle\langle\psi_{\mathrm{QD}}|$,
\begin{equation}
\begin{aligned}
\mathcal{W}_{\xi_C}^{C}\rho_{\mathrm{QD}}=\rho_{\mathrm{QD}}.
\end{aligned}
\end{equation}
The above relation shows that $\mathcal{W}_{\xi_C}^{C}$ is a strong magnetic $1$-form symmetry line within the full categorical $1$-form symmetry of the quantum-double model.

For $|C|>1$, the purely magnetic anyon line $(C,\mathbf{1})$ is non-invertible. When $C$ is a singleton conjugacy class contained in the center of $G$, it instead gives an invertible magnetic line. Moreover, the subset of purely magnetic anyon lines $\{F_{\xi}^{C}\}_{C\in\operatorname{Conj}(G)}$ is generally not closed under fusion. Fusing two purely magnetic anyon lines can produce sectors carrying nontrivial centralizer charge, including purely electric charges.

\textit{Decohered quantum double $\mathcal{D}(G)$} -- Consider the local decoherence described by the quantum channel $\mathcal{N} = \bigotimes_e \mathcal{N}_e$, such that
\begin{align}
    \mathcal{N}_e[\rho] = \frac{1}{|G|} \sum_{\Gamma \in \mathrm{Rep}(G)}\sum_{\alpha,\alpha'} d_\Gamma Z_{\Gamma,e,\alpha,\alpha'} \rho\left(Z_{\Gamma,e,\alpha,\alpha'}\right)^{\dagger},
    \label{Eq: decohered QD}
\end{align}
and 
\begin{equation}
    \begin{aligned}
        \cl=\N[\rho_{\mathrm{QD}}], \label{eq:rho_cl}
    \end{aligned}
\end{equation}
where $Z_{\Gamma,e,\alpha,\alpha'}$ is the $(\alpha,\alpha')$ matrix element of $Z_\Gamma$ acting on the edge $e$. The channel we considered in Eq.~\eqref{Eq: decohered QD} is the maximal charge-decoherence channel. More generally, one may replace $d_\Gamma$ by nonnegative parameters
$p_\Gamma\geq0$, such that
\begin{equation}
    \begin{aligned}
        \sum_{\Gamma \in \mathrm{Rep}(G)} p_{\Gamma} d_{\Gamma} = |G|.
    \end{aligned}
\end{equation}

These Kraus operators define a valid quantum channel (CPTP map), as it satisfies the following property:
\begin{equation}
    \begin{aligned}
        \sum_{\Gamma \in \mathrm{Rep}(G)} \frac{p_\Gamma \tr\left(Z_{\Gamma} Z^{\dagger}_{\Gamma}\right)}{|G|} =  \sum_{\Gamma \in \mathrm{Rep}(G)} \frac{p_\Gamma d_\Gamma}{|G|} = 1
    \end{aligned}
\end{equation}

We refer to this type of decoherence as the \textit{$Z$-type decoherence}. Physically, $Z_{\Gamma}$ creates electric charges in the quantum double, and the fusion of these charges generates the subcategory $\mathrm{Rep}(G)$ within the original anyon theory $\mathcal{Z}(\mathrm{Rep}(G))$.

Because electric charges braid trivially with each other, it is straightforward to verify that $\cl$ possesses strong 1-form symmetries $\mathcal{W}_{\xi_C}^\Gamma$ for any contractible closed ribbon $\xi_C$, namely
\begin{equation}
    \begin{aligned}
        \mathcal{W}_{\xi_C}^\Gamma\N_e[\rho_{\mathrm{QD}}]&= \N_e[\rho_{\mathrm{QD}}],
    \end{aligned}
\end{equation}
where we use the fact that $\mathcal{W}_{\xi_C}^\Gamma$ commutes with $Z_{\Gamma}$. Thus, the quantum-dimension-normalized operators $\{\mathcal W_{\xi_C}^\Gamma\}_{\Gamma\in\operatorname{Rep}(G)}$ remain strong $1$-form non-invertible symmetries of $\rho_{\mathrm{cl}}$.

However, under charge decoherence, the original strong magnetic $1$-form symmetry defined in Eq.~\eqref{eq:magnetic_symmetry} becomes weak. Specifically, we define the matrix elements of the quantum-dimension-normalized magnetic ribbon operator by
\begin{equation}
\begin{aligned}
\mathcal{W}_{\xi_C,i,i'}^{C} := \frac{1}{|C|}\sum_{k\in Z_C} F_{\xi_C}^{c_i^{-1},\,p_i k p_{i'}^{-1}}.
\end{aligned}
\end{equation}
The corresponding weak-symmetry action on a density matrix is defined as
\begin{equation}
\begin{aligned}
\widehat{\mathcal{W}}_{\xi_C}^{C}[\rho] := |C| \sum_{i,i'} \mathcal{W}_{\xi_C,i,i'}^{C}\, \rho\, \bigl(\mathcal{W}_{\xi_C,i,i'}^{C}\bigr)^{\dagger}.
\end{aligned}
\label{Eq: weak sym}
\end{equation}
Although a magnetic ribbon operator appears on both sides of the density matrix, the two ribbons should not be regarded as two disconnected loops. The sum over the common virtual indices produces a connected loop in the doubled virtual space. The two ribbon operators are supported on the same spatial ribbon $\xi_C$ and are joined at a chosen point. The resulting contracted action $\widehat{\mathcal{W}}_{\xi_C}^{C}$ is gauge invariant, so the choice of the joining point does not affect the weak-symmetry action. On the density matrix of the pure state quantum double, we have
\begin{equation}
    \begin{aligned}
        \widehat{\mathcal{W}}_{\xi_C}^{C}[\rho_{\mathrm{QD}}] &:= |C| \sum_{i,i'} \mathcal{W}_{\xi_C,i,i'}^{C}\, \rho_{\mathrm{QD}}\, \bigl(\mathcal{W}_{\xi_C,i,i'}^{C}\bigr)^{\dagger}\\
        &= \rho_{\mathrm{QD}}. \label{eq:weak_symmetry}
    \end{aligned}
\end{equation}

Maximal charge decoherence dephases the density matrix in the group basis, in which the magnetic ribbon operators act as permutations of the allowed configurations. When acting on $\rho_{\mathrm{cl}}$, $\widehat{W}^C_{\xi_C}$ commutes with each local channel $\mathcal{N}_e$, since the ket and bra have identical group configurations and the induced virtual representation matrices cancel between the two sides. The full dephasing channel $\mathcal{N}$ is covariant under the magnetic ribbon action, so the pure-state symmetry relation above gives
\begin{equation}
\begin{aligned}
\widehat{W}^C_{\xi_C}[\rho_{\mathrm{cl}}] &= \widehat{W}^C_{\xi_C}[\mathcal{N}[\rho_{\mathrm{QD}}]] \\
&= \mathcal{N}[\widehat{W}^C_{\xi_C}[\rho_{\mathrm{QD}}]] \\
&= \rho_{\mathrm{cl}}.
\end{aligned} \label{eq:weak_symmetry_2}
\end{equation}

The weak-symmetry action defined above is also compatible with the fusion algebra. Consider applying two weak-symmetry actions, labeled by the conjugacy classes $C$ and $D$, successively to a density matrix. We obtain
\begin{equation}
\begin{aligned}
&\quad \left(\widehat{\mathcal{W}}_{\xi_C}^{C}\circ\widehat{\mathcal{W}}_{\xi_C}^{D}\right)[\rho]\\
&=d_Cd_D\sum_{\substack{i,i'\\j,j'}}\mathcal{W}_{\xi_C,i,i'}^{C}\mathcal{W}_{\xi_C,j,j'}^{D}\,\rho\,\left(\mathcal{W}_{\xi_C,i,i'}^{C}\mathcal{W}_{\xi_C,j,j'}^{D}\right)^{\dagger},
\end{aligned}
\label{eq:sequential_weak_action}
\end{equation}
where $d_C=|C|$ and $d_D=|D|$ are the quantum dimensions of the purely magnetic anyons $(C,\mathbf{1})$ and $(D,\mathbf{1})$, respectively.
Alternatively, we may first fuse the two ribbon operators. Their fused virtual space decomposes into irreducible anyon sectors according to
\begin{equation}
\begin{aligned}
(C,\mathbf{1})\times(D,\mathbf{1})=\bigoplus_a N_{CD}^{a}\,a,
\end{aligned}
\end{equation}
where $a=(E,R)$ labels a general anyon sector of $D(G)$. In particular, $a$ is not restricted to a purely magnetic sector, since the subset of purely magnetic anyons is not generally closed under fusion. At the level of the quantum-dimension-normalized closed ribbon operators, this gives
\begin{equation}
\begin{aligned}
\mathcal{W}_{\xi_C}^{C}\mathcal{W}_{\xi_C}^{D}=\sum_a\frac{N_{CD}^{a}d_a}{d_Cd_D}\,\mathcal{W}_{\xi_C}^{a}.
\end{aligned}
\end{equation}
The sum over all virtual indices in Eq.~\eqref{eq:sequential_weak_action} is invariant under a change of basis in the fused virtual space. We may therefore evaluate it in the fusion basis, in which the product ribbon contains $N_{CD}^{a}$ copies of each sector $a$. It follows that
\begin{equation}
\begin{aligned}
&\quad \sum_{\substack{i,i'\\j,j'}}\mathcal{W}_{\xi_C,i,i'}^{C}\mathcal{W}_{\xi_C,j,j'}^{D}\,\rho\,\left(\mathcal{W}_{\xi_C,i,i'}^{C}\mathcal{W}_{\xi_C,j,j'}^{D}\right)^{\dagger}\\
&=\sum_a\frac{N_{CD}^{a}d_a^2}{d_C^2d_D^2}\sum_{m,m'}\mathcal{W}_{\xi_C,m,m'}^{a}\,\rho\,\bigl(\mathcal{W}_{\xi_C,m,m'}^{a}\bigr)^{\dagger}.
\end{aligned}
\label{eq:ribbon_fusion_decomposition}
\end{equation}
For a general anyon sector $a$ with quantum dimension $d_a$, we define the corresponding normalized two-sided ribbon action by
\begin{equation}
\begin{aligned}
\widehat{\mathcal{W}}_{\xi_C}^{a}[\rho]:=d_a\sum_{m,m'}\mathcal{W}_{\xi_C,m,m'}^{a}\,\rho\,\bigl(\mathcal{W}_{\xi_C,m,m'}^{a}\bigr)^{\dagger}.
\end{aligned}
\end{equation}
Substituting Eq.~\eqref{eq:ribbon_fusion_decomposition} into Eq.~\eqref{eq:sequential_weak_action}, we find
\begin{equation}
\begin{aligned}
&\quad \left(\widehat{\mathcal{W}}_{\xi_C}^{C}\circ\widehat{\mathcal{W}}_{\xi_C}^{D}\right)[\rho]\\
&=\frac{1}{d_Cd_D}\sum_a N_{CD}^{a}d_a^2\sum_{m,m'}\mathcal{W}_{\xi_C,m,m'}^{a}\,\rho\,\bigl(\mathcal{W}_{\xi_C,m,m'}^{a}\bigr)^{\dagger}\\
&=\sum_a\frac{N_{CD}^{a}d_a}{d_Cd_D}\,\widehat{\mathcal{W}}_{\xi_C}^{a}[\rho].
\end{aligned}
\label{eq:weak_action_fusion}
\end{equation}
The quantum dimensions satisfy
\begin{equation}
\begin{aligned}
d_Cd_D=\sum_a N_{CD}^{a}d_a,
\end{aligned}
\label{eq:fusion_ring}
\end{equation}
so the coefficients in Eq.~\eqref{eq:weak_action_fusion} sum to unity. Therefore, the normalized two-sided ribbon actions obey the normalized fusion algebra
\begin{equation}
\begin{aligned}
\widehat{\mathcal{W}}_{\xi_C}^{C}\circ\widehat{\mathcal{W}}_{\xi_C}^{D}=\sum_a\frac{N_{CD}^{a}d_a}{d_Cd_D}\,\widehat{\mathcal{W}}_{\xi_C}^{a}.
\end{aligned}
\end{equation}
Equation~\eqref{eq:weak_action_fusion} shows that the two procedures are equivalent: one may either apply the two weak-symmetry actions successively, or first fuse the two ribbon operators and then apply the resulting fused action to the density matrix. In particular, since
\begin{equation}
\begin{aligned}
\widehat{\mathcal{W}}_{\xi_C}^{C}[\rho_{\mathrm{cl}}]=\widehat{\mathcal{W}}_{\xi_C}^{D}[\rho_{\mathrm{cl}}]=\rho_{\mathrm{cl}},
\end{aligned}
\end{equation}
both procedures give the same resulting density matrix,
\begin{equation}
\begin{aligned}
\left(\widehat{\mathcal{W}}_{\xi_C}^{C}\circ\widehat{\mathcal{W}}_{\xi_C}^{D}\right)[\rho_{\mathrm{cl}}]&=\sum_a\frac{N_{CD}^{a}d_a}{d_Cd_D}\,\widehat{\mathcal{W}}_{\xi_C}^{a}[\rho_{\mathrm{cl}}]\\
&=\rho_{\mathrm{cl}},
\end{aligned}
\end{equation}
where we use Eqs.~\eqref{eq:rho_cl}, \eqref{eq:weak_symmetry_2}, and \eqref{eq:fusion_ring} in the derivation.
In summary, we have shown that the maximally charge-decohered quantum-double state $\rho_{\mathrm{cl}}$ possesses a strong $\operatorname{Rep}(G)$-type $1$-form symmetry and a weak magnetic non-invertible $1$-form symmetry. Specifically, for any contractible closed ribbon $\xi_C$, we have
\begin{equation}
\begin{aligned}
\mathcal{W}_{\xi_C}^{\Gamma}\rho_{\mathrm{cl}}=\rho_{\mathrm{cl}},\qquad \widehat{\mathcal{W}}_{\xi_C}^{C}\!\left[\rho_{\mathrm{cl}}\right]=\rho_{\mathrm{cl}}.
\end{aligned}
\end{equation}
where $\Gamma\in\operatorname{Rep}(G)$ and $C$ labels a conjugacy class of $G$. For a noncontractible closed ribbon $\xi_C$, the action of these symmetries depends on the global holonomy sector along $\xi_C$, and the above invariance relations need not hold for an arbitrary $\rho_{\mathrm{cl}}$.

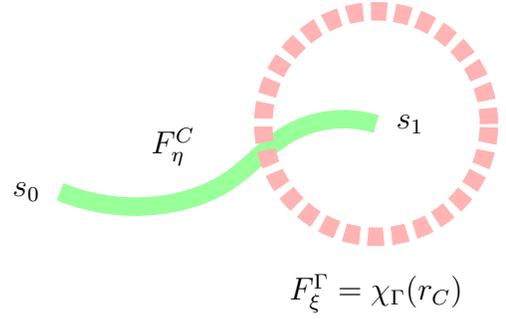
\begin{figure}
\begin{tikzpicture}[use Hobby shortcut,xscale=1.5,yscale=1.5]
\tikzstyle{sergio}=[rectangle,draw=none]
\coordinate[] (A6) at (1,0.7);
\coordinate[] (A7) at (2,0.6);
\coordinate[] (A8) at (2.5,0.8);
\coordinate[] (A9) at (3,1.2);
\coordinate[] (A10) at (3.8,1.3);
\draw[color=green!40,line width=7pt] (A6) .. (A7) .. (A8) .. (A9) .. (A10);
\path (4.1,1.3) node [style=sergio]{\large$s_1$};
\path (0.7,0.7) node [style=sergio]{\large$s_0$};
\draw[red!30, line width=7pt, dash pattern={on 6pt off 3pt}] (3.8,1.3)circle (1);
\path (2,1.1) node [style=sergio]{\large$\mathcal{W}_{\eta}^C$};
\path (3.8, -0.2) node [style=sergio]{\large$\mathcal{W}_{\xi_C}^{\Gamma}$};
\path (5.5, 0.2) node [style=sergio]{\large $=\chi_{\Gamma}(r_C)/d_{\Gamma}$};
\end{tikzpicture}
\caption{'t Hooft anomaly/SSB of 1-form ribbon symmetries. The open green ribbon depicts $\mathcal{W}_{\eta}^C$, the closed dashed red ribbon depicts $\mathcal{W}_{\xi_C}^{\Gamma}$, and $r_C\in C$.}
\label{Fig: Ribbons}
\end{figure}

\textit{Spontaneous breaking of 1-form non-invertible symmetry} -- For topological orders, topological degeneracy arises from the nontrivial braiding statistics of anyons. In $(2+1)$ dimensions, a $1$-form symmetry operator is supported on a closed one-dimensional loop in a spatial slice, and an anyon line with nontrivial braiding with this loop serves as its charged order parameter. In this section, we demonstrate that the nontriviality of $\cl$ is captured by the SSB of the non-invertible 1-form symmetries defined in Eq.~\eqref{Eq: strong sym} and Eq.~\eqref{Eq: weak sym}.
The strong electric and weak magnetic 1-form symmetries exhibit a mixed 't Hooft anomaly. If the closed ribbon operator $\widehat{\mathcal{W}}^{C}_{\xi_C}\left[\cdot\right]$ is truncated to an open ribbon $\eta$ with endpoints $s_0$ and $s_1$, the resulting open ribbon operator $\widehat{\mathcal{W}}^{C}_{\eta}\left[\cdot\right]$ creates magnetic anyons at $s_0$ and $s_1$ on both sides of the density matrix. A closed electric ribbon operator $\mathcal{W}_{\xi_C}^\Gamma$ enclosing exactly one endpoint of $\eta$ can then detect the presence of these excitations (see Fig.~\ref{Fig: Ribbons}), namely
\begin{equation}
\begin{aligned}
\mathcal W_{\xi_C}^\Gamma \left(\widehat{\mathcal{W}}^{C}_{\eta}\left[\cl\right]\right) =
\frac{\chi_\Gamma(r_C)}{d_\Gamma}
\widehat{\mathcal{W}}^C_{\eta}\left[\mathcal W_{\xi_C}^\Gamma\rho_{\mathrm{cl}}\right].
\end{aligned}
\label{Eq: anomaly}
\end{equation}
For $\chi_\Gamma(r_C)/d_\Gamma\neq 1$, Eq.~\eqref{Eq: anomaly} shows the mixed 't Hooft anomaly between the strong electric and weak magnetic 1-form symmetries. Treating $\mathcal W_{\xi_C}^\Gamma$ as a strong symmetry, we obtain
\begin{align}
\mathcal{F}\left(\cl,\widehat{\mathcal{W}}^C_{\xi_C}\left[\cl\right]\right)=1,
\end{align}
where $\xi_C$ is a closed contractible ribbon of arbitrary linear size $R$, and for a nontrivial conjugacy class $C\neq\{e\}$,
\begin{equation}
\begin{aligned}
\tr\left(\cl \mathcal{W}^C_{\xi_C}\right) = 0.
\end{aligned}
\end{equation}
Thus, the charged magnetic ribbon has long-range fidelity order but no conventional loop order, which implies the SWSSB of the strong electric $1$-form symmetry $\mathcal{W}_{\xi_C}^{\Gamma}$, with $\widehat{\mathcal{W}}_{\xi_C}^{C}\left[\cdot\right]$ serving as the fidelity ribbon order parameter.

In abelian quantum double models, both $Z$ and $X$ types of errors can induce the SWSSB of 1-form symmetries. For a non-abelian group, however, generic $X$-type decoherence breaks the electric $\operatorname{Rep}(G)$ symmetry MPOs $\mathcal W_{\xi_C}^\Gamma$. Among the magnetic group-valued lines, only those labeled by the central elements $g\in Z(G)$ commute with all Kraus operators and remain strong.

Consider the $X$-type decoherence described by the quantum channel $\N'=\bigotimes_{e}\N_e'$, and
\begin{align}
\N_e'[\rho]=\frac{1}{|G|}\sum_{g\in G}L_{g,e}^{+}\rho(L_{g,e}^+)^\dag.
\end{align}
It is straightforward to see that the ribbon operator in Eq.~\eqref{Eq: strong sym} fails to commute with the Kraus operator $L_{g,e}^+$. For Eq.~\eqref{eq:magnetic_symmetry}, only those ribbon operators associated with conjugacy classes in the center of the group retain strong symmetry, while all others are destroyed. Thus, generic $X$-type decoherence leaves only the invertible magnetic strong symmetries associated with the group center and no generic non-invertible strong symmetry.

Physically, when the Kraus operators are chosen as local anyon creation operators, analogous to anyon condensation, the quantum channel incoherently proliferates the corresponding anyons~\cite{Ellison_2025}. If the anyons generated by these Kraus operators form a subcategory of $\mathcal{Z}(\mathrm{Rep}(G))$ in which all corresponding anyons braid trivially, then the associated strong symmetry is preserved. On the other hand, generic $X$-type operators in non-abelian quantum double generate superpositions of magnetic anyons, which in general do not form a subcategory. Consequently, under generic $X$-type decoherence, the corresponding strong symmetries are broken.

In the charge-decoherence example, the electric charges form the canonical maximal Lagrangian symmetric fusion subcategory $\mathrm{Rep}(G)\subset\mathcal{Z}\left(\mathrm{Rep}(G)\right)$. For non-abelian $G$, this subcategory contains non-invertible anyons, and maximal decoherence generated by this subcategory produces a density matrix exhibiting non-invertible $1$-form SWSSB. Moreover, $\mathcal{Z}\left(\mathrm{Rep}(G)\right)$ may admit {\it braided auto-equivalences}, which permute the anyon types while preserving their fusion and braiding structures~\cite{kitaev2012models,etingof2010fusion,nikshych2014categorical,aasen2022adiabatic,li2026anyon}. Applying a braided auto-equivalence to the canonical charge subcategory generates another maximal Lagrangian symmetric fusion subcategory, and maximal decoherence generated by the resulting subcategory likewise produces an SWSSB density matrix. 

As a toy example, consider the $S_3$ quantum double. As illustrated in the Supplementary Material, the vacuum and all electric charges form the subcategory $\{A,B,C\}$, where $B$ is an abelian charge and $C$ is a non-abelian charge. The $S_3$ quantum double admits a nontrivial braided auto-equivalence that exchanges $C$ with the non-abelian flux $F$, thereby mapping $\{A,B,C\}$ to $\{A,B,F\}$. Therefore, maximal decoherence generated by either $\{A,B,C\}$ or $\{A,B,F\}$ produces a density matrix exhibiting non-invertible $1$-form SWSSB, and these two choices exhaust the eligible maximal Lagrangian symmetric fusion subcategories in the $S_3$ quantum double.

\textit{Locally indistinguishable set} -- Subsequently, we show that the decohered states $\{\cl\}$ form a \textit{locally indistinguishable set} \cite{shi2020fusion,sang2025mixed}. The formal definition of the locally indistinguishable set is reviewed as follows

\begin{figure}
\begin{tikzpicture}[scale=0.9]
\tikzstyle{sergio}=[rectangle,draw=none]
\filldraw[fill=green!20, draw=black, thick] (0,0)circle (2);
\filldraw[fill=violet!20, draw=black, thick] (0,0)circle (1);
\path (0,0) node [style=sergio,color=red] {\large$A$};
\path (1.5,0.25) node [style=sergio,color=red] {$w(B)$};
\path (-1.075,1.075) node [style=sergio,color=red] {\large$B$};
\path (-1.75,1.75) node [style=sergio,color=red] {\large$C$};
\path (5,0) node [style=sergio] {\large $I_\rho(A:C|B)\sim e^{-w(B)/\xi_{\mathrm{M}}}$};
\draw[thick, <->] (1,0) -- (2,0);
\end{tikzpicture}
\caption{A tripartition $A$, $B$, and $C$ of the system. Here $B$ is the buffer region of $A$ with the width $w(B)$.}
\label{Fig: tripartition}
\end{figure}
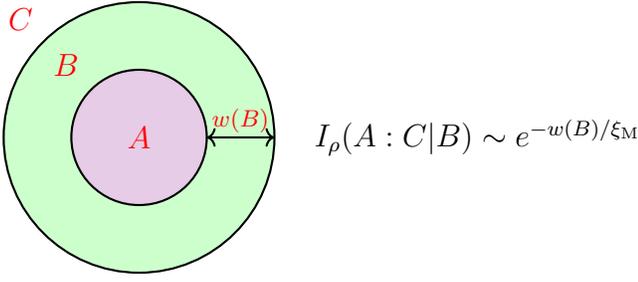

\begin{dfn}[Locally indistinguishable set]
For a mixed state $\rho$ defined on a region $\Sigma$, its locally indistinguishable set $\mathcal{Q}(\rho;\xi_0)$ is a set of density matrices that are identical to $\rho$ in any simply connected subregion $\forall A\subset\Sigma$ and have a Markov length no greater than $\xi_0$:
\begin{align}
\mathcal{Q}(\rho;\xi_0)=\left\{\sigma\big|\xi(\sigma)\leq\xi_0;\sigma_A=\rho_A\right\}.
\end{align}
\end{dfn}

\begin{dfn}[Markov length]
    The Markov length of $\rho$ is defined by the exponentially decaying conditional mutual information (CMI) as
\begin{align}
I_\rho(A:C|B)&=S(AB)+S(BC)-S(B)-S(ABC)\nonumber\\
&\sim e^{-w(B)/\xi_{\mathrm{M}}},
\end{align}
where $A\cup B \cup C$ is a tripartition of the system in which $B$ is a buffer region of width $w(B)$ (see Fig. \ref{Fig: tripartition}), $S(X)$ denotes the von Neumann entropy of the reduced density matrix on region $X$, and $\xi_{\mathrm{M}}$ is the Markov length.
\end{dfn}

For the maximally decohered fixed-point density matrices considered below, we use the shorthand $\mathcal{Q}\{\rho_{\mathrm{cl}}\}:=\mathcal{Q}(\rho_{\mathrm{cl}};\xi_0)$, where $\xi_0$ is fixed to a finite value.

\begin{thm} \label{thm_1}
The maximally decohered density matrices in $\mathcal{Q}\{\rho_{\mathrm{cl}}\}$, obtained from the ground states of Kitaev's quantum double $\mathcal{D}(G)$, form a locally indistinguishable set: for an arbitrary simply connected subregion $A\subset\Sigma$, and two mixed states $\cl^1,\cl^2\in\mathcal{Q}\{\cl\}$, we have
\begin{equation}
\begin{aligned}
\rho_{\mathrm{cl},A}:=\tr_{\bar A}\left(\rho_{\mathrm{cl}}^1\right)=\tr_{\bar A}\left(\rho_{\mathrm{cl}}^2\right),
\end{aligned}
\end{equation}
and for an arbitrary mixed state $\sigma\in\mathcal{Q}\{\cl\}$, the respective Markov length $\xi_{\mathrm{M}}$ is finite:
\begin{align}
I_{\sigma}(A:C|B)\sim e^{-w(B)/\xi_{\mathrm{M}}}.
\end{align}
\end{thm}

We can further show that the number of extremal points in $\mathcal{Q}\{\cl\}$ is equal to the ground state degeneracy of the quantum-double Hamiltonian in Eq.~\eqref{Eq: Hamiltonian} before decoherence. 
\begin{thm} \label{thm_2}
The set of the extremal points of $\mathcal{Q}\{\rho_{\mathrm{cl}}\}$, denoted by $\mathrm{Ext}(\mathcal{Q}\{\rho_{\mathrm{cl}}\})$, is isomorphic to the $G$-orbits of $\mathrm{Hom}(\pi_1(\Sigma),G)$ under conjugation, namely
\begin{align}
\mathrm{Ext}(\mathcal{Q}\{\rho_{\mathrm{cl}}\}) \cong \mathrm{Hom}(\pi_1(\Sigma),G)/\mathrm{Ad}_G,
\end{align}
where $\pi_1(\Sigma)$ is the fundamental group of $\Sigma$, and
\begin{align}
    \mathrm{Ad} : G \to \mathrm{Aut} (G),\quad \mathrm{Ad}_g(x):= g x \bar{g}.
\end{align}
\end{thm}
This theorem shows that the quantum information encoded in the ground-state subspace of the quantum double model decoheres into classical information under the $Z$-type errors defined in Eq.~\eqref{Eq: decohered QD}, and this classical information is stored in the convex set $\mathcal{Q}\{\cl\}$. Each extremal point of $\mathcal{Q}\{\cl\}$ corresponds to a maximally decohered density matrix associated with a distinct global holonomy sector. 

The proofs of both theorems are provided in the Supplementary Materials \cite{supp}.

For completeness, we further show that the locally indistinguishable set $\mathcal{Q}\{\cl\}$ forms an information convex set, as previously discussed in Ref.~\cite{sang2025mixed}.

\begin{corollary} \label{corollary_3}
The locally indistinguishable set $\mathcal{Q}\{\cl\}$ defines an information convex set, i.e., if $\sigma_1,\sigma_2\in\mathcal{Q}\{\cl\}$, then $\forall p\in[0,1]$, $\sigma_3=p\sigma_1+(1-p)\sigma_2\in\mathcal{Q}\{\cl\}$. 
\end{corollary}

\begin{proof}
It is sufficient to show that: 
\begin{enumerate}[1.]
\item $\sigma_3$ is locally indistinguishable from $\cl\big|_A=\tr_{\bar{A}}(\cl)$;
\item $\sigma_3$ has a finite Markov length.
\end{enumerate}
The first statement follows from
\begin{align}
\sigma_{3,A}=p\sigma_{1,A}+(1-p)\sigma_{2,A}=\cl\big|_A.
\end{align}
For the second statement, the CMI can be expressed as
\begin{align}
I_{\sigma_3}(A:C|B)=S_{\sigma_3}(A|B)-S_{\sigma_3}(A|BC),
\end{align}
where $S(X|Y)=S(XY)-S(Y)$ is the quantum conditional entropy. Since $A\cup B$ is a simply connected region, local indistinguishability implies $S_{\sigma_3}(A|B)=S_{\sigma_{1}}(A|B) = S_{\sigma_{2}}(A|B)$. 

For the quantum conditional entropy $S_{\sigma_3}(A|BC)$, we apply the concavity of quantum conditional entropy, namely
\begin{align}
S_{\sigma_3}(A|BC)\geq pS_{\sigma_1}(A|BC)+(1-p)S_{\sigma_2}(A|BC),
\end{align}
which gives an upper bound of the CMI, namely
\begin{equation}
    \begin{aligned}
        I_{\sigma_3}(A:C |B) &\leq pI_{\sigma_1}(A:C|B) \\
        &+(1-p)I_{\sigma_2}(A:C|B).
    \end{aligned}
\end{equation}
Therefore, we conclude that the Markov length of $\sigma_3$ is upper bounded by the finite value $\max\{\xi_{\sigma_1},\xi_{\sigma_2}\}$.
\end{proof}

\textit{Conclusion and outlook} -- In this work, we systematically study Kitaev’s quantum double model for an arbitrary finite group $G$ under certain types of local decoherence, together with the resulting spontaneous symmetry-breaking patterns. We show that under maximal $Z$-type decoherence, the ground state codespace $\mathcal{C}[\mathcal{D}(G)]$ is mapped to a locally indistinguishable set $\mathcal{Q}\{\cl\}$. We further define strong and weak 1-form non-invertible symmetries of the decohered density matrices $\cl$ via closed ribbon operators, and demonstrate that $\cl$ exhibits SWSSB of the strong 1-form symmetry and WTSSB of the weak 1-form symmetry. Finally, we prove that the locally indistinguishable set $\mathcal{Q}\{\cl\}$ indeed forms an information convex set. A basis of ground states labeled by definite global holonomy sectors is mapped to the distinct extremal points of the information convex set $\mathcal{Q}\{\cl\}$, while a general coherent superposition of these states is mapped to a classical convex combination of the extremal points. This highlights that the quantum information originally stored in $\mathcal{C}[\mathcal{D}(G)]$ is decohered into classical information stored in $\mathcal{Q}\{\cl\}$.

We conclude this work with several open questions. We have shown that maximal $Z$-type decoherence maps the quantum information stored in the ground-state subspace $\mathcal{C}[\mathcal{D}(G)]$ to classical information stored in the locally indistinguishable convex set $\mathcal{Q}\{\cl\}$. A key open problem is to determine the error threshold of this decodability transition in Kitaev’s quantum double model under $Z$-type noise~\cite{workinprogress}. We conjecture that for an error channel $\mathcal{E}$ with parameter $p$ below a critical value $p_c$, the error is correctable in the sense that there exists a quasi-local, low-depth recovery channel $\tilde{\mathcal{E}}$ such that $\tilde{\mathcal{E}}\circ\mathcal{E}[\rho]\simeq\rho$. By contrast, for $p>p_c$, no such recovery channel exists.

Another open problem is to extend the paradigm of SWSSB to string-net models~\cite{PhysRevB.71.045110,PhysRevB.103.195155}. For instance, one may consider decohering the doubled Ising anyon theory by incoherently proliferating the $\psi \bar{\psi}$ anyon, as discussed in Ref.~\cite{Ellison_2025}. Following our analysis, the resulting decohered density matrix exhibits strong symmetries associated with the $\psi$, $\bar{\psi}$, $\psi\bar{\psi}$, and $\sigma \bar{\sigma}$ anyons, while the remaining 1-form symmetries become weak. Nevertheless, a systematic understanding of SWSSB in string-net models is still lacking, and we leave a detailed study of this problem to future work.

\begin{acknowledgments}
\textit{Acknowledgments} -- We thank Zhen Bi, Isaac Kim, Andy Lucas, and Ruochen Ma for insightful discussions. ZS thanks Shengqi Sang, Timothy Hsieh, and Isaac Kim for their collaboration on a related project. ZS is supported by the National Science Foundation under Grant No. DRL I-TEST 2148467. JHZ is supported by the US Department of Energy under Grant No. DE-SC0024324. This research was supported in part by a grant NSF PHY-2309135 to the Kavli Institute for Theoretical Physics (KITP). 
\end{acknowledgments}

\bibliography{Refs}

\newpage
\appendix

\onecolumngrid

\section{Properties of generalized Pauli operators}
The commutation relations of $L_h^{+}$, $L_h^{-}$, and $Z_\Gamma$ are summarized as follows
\begin{align}
\begin{gathered}
\relax[L_h^{+},L_g^{-}]=0\\
L_g^{+}Z_\Gamma = \Gamma(\bar{g})Z_\Gamma L_g^{+}\\
L_g^{-} Z_\Gamma = Z_\Gamma\Gamma(g) L_g^{-}
\end{gathered}, \label{eq:conmmutation}
\end{align}

The validity of $B_p$ terms in terms of the generalized Pauli $Z$ operators is ensured by the great orthogonality theorem of group representation theory: 
\begin{thm}[Great orthogonality theorem]
For a group $G$ with all irreps $\mathrm{Rep(G)}$, we have
\begin{align}
\frac{1}{|G|}\sum_{\Gamma\in\mathrm{Rep}(G)}\chi_\Gamma^*(g_i)\chi_\Gamma(g_j)=\frac{\delta_{[g_i],[g_j]}}{|g_i|},
\end{align}
where $\chi_\Gamma$ is the character of the irrep $\Gamma$, $|g_i|$ is the total number of elements in the conjugacy class of $g_i\in G$. In particular, if we take $g_i=e$ as the identity element of $G$, we get
\begin{align}
\frac{1}{|G|}\sum_{\Gamma\in\mathrm{Rep}(G)}d_\Gamma\cdot\chi_\Gamma(g_j)=\delta_{e,g_j},
\end{align}
where $d_\Gamma$ is the dimension of the irrep $\Gamma$.
\label{Thm: GOT}
\end{thm}

\section{Deformation of ribbon operators}

We summarize the algebra between the ribbon operator and the Hamiltonian terms here~\cite{beigi2010quantum-2ab,Bombin2007}. The multiplication of ribbon operators on the same ribbon $\xi$ is
\begin{align}
    F^{h_1, g_1}_{\xi} F^{h_2, g_2}_{\xi} = \delta_{g_1, g_2} F^{h_1 h_2, g_2}_{\xi}. \label{eq:ribbon_multiplication}
\end{align}
If the end of a ribbon $\xi_1$ is the start of another ribbon $\xi_2$, we can denote the composition of the two ribbons as $\xi = \xi_1 \xi_2$. The ribbon operator $F^{h,g}_{\xi}$ on ribbon $\xi$ obeys the co-multiplication rule
\begin{align}
    F^{h,g}_{\xi} = \sum_{k \in G} F^{h,k}_{\xi_1} F^{\bar{k} h k, \bar{k} g}_{\xi_2}. \label{eq:ribbon_comultiplication}
\end{align}

At the beginning of the ribbon operator $s_0$, we have the following algebra.
\begin{equation}
\begin{aligned}
    A^k_{s_0} F^{h,g}_{\xi} &= F^{k h \bar{k}, kg}_{\xi} A^k_{s_0}, \\
    B_{s_0}^k F^{h,g}_{\xi} &= F^{h,g}_{\xi} B_{s_0}^{kh}.
\end{aligned}
\end{equation}
At the end of the ribbon operator $s_1$, we have
\begin{equation}
    \begin{aligned}
        A^k_{s_1} F^{h,g}_{\xi} &= F^{h,g \bar{k}}_{\xi} A_{s_1}^k, \\
        B_{s_1}^{k} F^{h,g}_{\xi} &= F^{h,g}_{\xi} B_{s_1}^{\bar{g} \bar{h} g k}.
    \end{aligned}
\end{equation}

The ribbon operator $F^{h,g}_{\xi}$ can be deformed by applying local $A_v^k$ or $B_p$ operators. Consider the following ribbon operator,

\begin{equation}
\begin{tikzpicture}[scale=1.25]
\tikzstyle{sergio}=[rectangle,draw=none]
\filldraw[fill=gray!30, draw=gray, dotted]
(-4,-1) -- (-2,-1) -- (-2.5,-1.5) -- (-4.5,-1.5) -- cycle;
\filldraw[fill=gray!50, draw=gray, dotted]
(-2,-1) -- (-1,-1) -- (-1.5,-1.5) -- (-2.5,-1.5) -- cycle;
\filldraw[fill=gray!70, draw=gray, dotted]
(-1,-1) -- (0,-1) -- (-0.5,-1.5) -- (-1.5,-1.5) -- cycle;

\foreach \x in {-4,-3,...,-1} {
  \draw (\x,-1)--(\x+1,-1) node[currarrow,pos=0.5,sloped]{};
}

\foreach \x in {-4,-3,...,-1} {
  \draw (\x,-2)--(\x+1,-2);
}

\foreach \x in {-4,-3,...,0} {
  \draw (\x,-1)--(\x,-2) node[currarrow,pos=0.5,sloped]{};
}

\path (-5,-1.5) node [style=sergio]{$F^{h,g}_{\xi}$};

\filldraw[red] (-2,-1) circle (2pt);
\path (-2,-0.8) node [style=sergio]{$v$};

\path (-3.5,-0.8) node [style=sergio]{$x_1$};
\path (-2.5,-0.8) node [style=sergio]{$x_2$};
\path (-1.5,-0.8) node [style=sergio]{$x_3$};
\path (-0.5,-0.8) node [style=sergio]{$x_4$};

\path (-4,-1.75) node [style=sergio]{$y_1$};
\path (-3,-1.75) node [style=sergio]{$y_2$};
\path (-2,-1.75) node [style=sergio]{$y_3$};
\path (-1,-1.75) node [style=sergio]{$y_4$};

\path (-3.3,-1.3) node [style=sergio]{$\xi_1$};
\path (-1.7,-1.3) node [style=sergio]{$\xi_2$};
\path (-0.7,-1.3) node [style=sergio]{$\xi_3$};

\filldraw[fill=gray!30, draw=gray, dotted]
(-4,-3) -- (-2,-3) -- (-2.5,-3.5) -- (-4.5,-3.5) -- cycle;
\filldraw[fill=gray!50, draw=gray, dotted]
(-2,-3) -- (-1,-3) -- (-1.5,-3.5) -- (-2.5,-3.5) -- cycle;
\filldraw[fill=gray!70, draw=gray, dotted]
(-1,-3) -- (0,-3) -- (-0.5,-3.5) -- (-1.5,-3.5) -- cycle;

\foreach \x in {-4,-3,...,-1} {
  \draw (\x,-3)--(\x+1,-3) node[currarrow,pos=0.5,sloped]{};
}

\foreach \x in {-4,-3,...,-1} {
  \draw (\x,-4)--(\x+1,-4);
}

\foreach \x in {-4,-3,...,0} {
  \draw (\x,-3)--(\x,-4) node[currarrow,pos=0.5,sloped]{};
}

\path (-5,-3.5) node [style=sergio]{$= \delta_{\hat{x}_4,g}$};

\filldraw[red] (-2,-3) circle (2pt);
\path (-2,-2.8) node [style=sergio]{$v$};

\path (-3.5,-2.8) node [style=sergio]{$x_1$};
\path (-2.5,-2.8) node [style=sergio]{$x_2$};
\path (-1.5,-2.8) node [style=sergio]{$x_3$};
\path (-0.5,-2.8) node [style=sergio]{$x_4$};

\path (-4,-3.75) node [style=sergio]{$h y_1$};
\path (-3,-3.75) node [style=sergio]{$\hat{\bar{x}}_1 h \hat{x}_1 y_2$};
\path (-1.8,-3.75) node [style=sergio]{$\hat{\bar{x}}_2 h \hat{x}_2 y_3$};
\path (-0.6,-3.75) node [style=sergio]{$\hat{\bar{x}}_3 h \hat{x}_3 y_4$};

\path (-3.3,-3.3) node [style=sergio]{$\xi_1$};
\path (-1.7,-3.3) node [style=sergio]{$\xi_2$};
\path (-0.7,-3.3) node [style=sergio]{$\xi_3$};

\end{tikzpicture}
\end{equation}
where $\xi = \xi_1 \xi_2 \xi_3$, $\hat{x}_i = x_1 x_2 ... x_i$, and $\hat{\bar{x}}_i = (x_1 x_2 ... x_i)^{-1}$.

We can apply the co-multiplication rule to decompose it into the following form,
\begin{equation}
    \begin{aligned}
        F^{h,g}_{\xi} = \sum_{k,l} F^{h,k}_{\xi_1} F^{k^{-1} h k,l}_{\xi_2} F^{(kl)^{-1} h (kl),(kl)^{-1} g}_{\xi_3},
    \end{aligned}
\end{equation}
where $k = \hat{x}_2$ and $l = x_3$.

Consider the $A^{\bar{h}}_v = F^{\bar{h},e}_{\xi'}$ term in the Hamiltonian

\begin{equation}
\begin{tikzpicture}[scale=1.3]
\tikzstyle{sergio}=[rectangle,draw=none]

\filldraw[fill=gray!50, draw=gray, dotted]
(-0.5,0.5) -- (0,0) -- (-0.5,-0.5) -- cycle;
\filldraw[fill=gray!50, draw=gray, dotted]
(-0.5,0.5) -- (0.5,0.5) -- (0,0) -- cycle;
\filldraw[fill=gray!50, draw=gray, dotted]
(0.5,0.5) -- (0,0) -- (0.5,-0.5) -- cycle;
\filldraw[fill=gray!50, draw=gray, dotted]
(0.5,-0.5) -- (0,0) -- (-0.5,-0.5) -- cycle;

\draw (-1,0)--(0,0) node[currarrow,pos=0.5,sloped]{};
\draw (0,0)--(1,0) node[currarrow,pos=0.5,sloped]{};
\draw (0,0)--(0,-1) node[currarrow,pos=0.5,sloped]{};
\draw (0,1)--(0,0) node[currarrow,pos=0.5,sloped]{};

\path (-1.4,0) node [style=sergio]{$F^{\bar{h}, e}_{\xi'}$};
\path (-0.5,0.2) node [style=sergio]{$x_2$};
\path (0.5,0.2) node [style=sergio]{$x_3$};
\path (-0.2,0.6) node [style=sergio]{$y_5$};
\path (-0.2,-0.4) node [style=sergio]{$y_3$};
\path (1.5,0) node [style=sergio]{$=$};
\path (0.2,-0.2) node [style=sergio]{$\xi'$};

\filldraw[fill=gray!50, draw=gray, dotted]
(2.5,0.5) -- (3,0) -- (2.5,-0.5) -- cycle;
\filldraw[fill=gray!50, draw=gray, dotted]
(2.5,0.5) -- (3.5,0.5) -- (3,0) -- cycle;
\filldraw[fill=gray!50, draw=gray, dotted]
(3.5,0.5) -- (3,0) -- (3.5,-0.5) -- cycle;
\filldraw[fill=gray!50, draw=gray, dotted]
(3.5,-0.5) -- (3,0) -- (3-0.5,-0.5) -- cycle;

\draw (2,0)--(3,0) node[currarrow,pos=0.5,sloped]{};
\draw (3,0)--(4,0) node[currarrow,pos=0.5,sloped]{};
\draw (3,0)--(3,-1) node[currarrow,pos=0.5,sloped]{};
\draw (3,1)--(3,0) node[currarrow,pos=0.5,sloped]{};

\path (2.5,0.2) node [style=sergio]{$x_2 h$};
\path (3.5,0.2) node [style=sergio]{$\bar{h} x_3$};
\path (3-0.2,0.6) node [style=sergio]{$y_5 h$};
\path (3-0.2,-0.4) node [style=sergio]{$\bar{h} y_3$};

\path (3.2,-0.2) node [style=sergio]{$\xi'$};
\filldraw[red] (0,0) circle (2pt);
\path (0.1,0.15) node [style=sergio]{$v$};
\filldraw[red] (3,0) circle (2pt);
\path (3.1, 0.15) node [style=sergio]{$v$};

\end{tikzpicture}
\end{equation}
where $\xi'$ denotes the ribbon winding around the vertex $v$ in the counterclockwise direction.

It's free to insert this operator into the above ribbon operator at vertex $v$, since the operator serves as a gauge transformation. We have
\begin{equation}
    \begin{aligned}
        &\quad F^{\bar{k} \bar{h} k, e}_{\xi'} F^{h,g}_{\xi} \\
        &= \sum_{k,l} F^{h,k}_{\xi_1} F^{\bar{k} \bar{h} k, e}_{\xi'} F^{\bar{k} h k, l}_{\xi_2} F^{\overline{(kl)} h (kl), \overline{(kl)} g}_{\xi_3}.
    \end{aligned}
\end{equation}

Near $v$, the operator acts as follows,
\begin{equation}
\begin{tikzpicture}[scale=1.2]
\tikzstyle{sergio}=[rectangle,draw=none]

\filldraw[fill=gray!50, draw=gray, dotted]
(-0.5,0.5) -- (0,0) -- (-0.5,-0.5) -- cycle;
\filldraw[fill=gray!50, draw=gray, dotted]
(-0.5,0.5) -- (0.5,0.5) -- (0,0) -- cycle;
\filldraw[fill=gray!50, draw=gray, dotted]
(0.5,0.5) -- (0,0) -- (0.5,-0.5) -- cycle;
\filldraw[fill=gray!50, draw=gray, dotted]
(0.5,-0.5) -- (0,0) -- (-0.5,-0.5) -- cycle;

\draw (-1,0)--(0,0) node[currarrow,pos=0.5,sloped]{};
\draw (0,0)--(1,0) node[currarrow,pos=0.5,sloped]{};
\draw (0,0)--(0,-1) node[currarrow,pos=0.5,sloped]{};
\draw (0,1)--(0,0) node[currarrow,pos=0.5,sloped]{};

\path (-1.5,0) node [style=sergio]{$F^{\bar{k} \bar{h} k, e}_{\xi'} F^{\bar{k} h k, l}_{\xi_2}$};
\path (-0.5,0.2) node [style=sergio]{$x_2$};
\path (0.5,0.2) node [style=sergio]{$x_3$};
\path (-0.2,0.6) node [style=sergio]{$y_5$};
\path (-0.2,-0.4) node [style=sergio]{$y_3$};
\path (1.6,0) node [style=sergio]{$= F^{\bar{k} \bar{h} k, e}_{\xi'}$};
\path (0.2,-0.2) node [style=sergio]{$\xi'$};

\filldraw[fill=gray!50, draw=gray, dotted]
(2.5,-2.5+3) -- (3,-3+3) -- (2.5,-3.5+3) -- cycle;
\filldraw[fill=gray!50, draw=gray, dotted]
(2.5,-2.5+3) -- (3.5,-2.5+3) -- (3,-3+3) -- cycle;
\filldraw[fill=gray!50, draw=gray, dotted]
(3.5,-2.5+3) -- (3,-3+3) -- (3.5,-3.5+3) -- cycle;
\filldraw[fill=gray!50, draw=gray, dotted]
(3.5,-3.5+3) -- (3,-3+3) -- (3-0.5,-3.5+3) -- cycle;

\draw (2,0)--(3,0) node[currarrow,pos=0.5,sloped]{};
\draw (3,0)--(4,0) node[currarrow,pos=0.5,sloped]{};
\draw (3,0)--(3,-1) node[currarrow,pos=0.5,sloped]{};
\draw (3,1)--(3,0) node[currarrow,pos=0.5,sloped]{};

\path (2.5,0.2) node [style=sergio]{\small $x_2$};
\path (3.5,0.2) node [style=sergio]{\small $x_3$};
\path (3-0.2,0.6) node [style=sergio]{\small $y_5$};
\path (3-0.4,-0.4) node [style=sergio]{\small $\bar{k} h k y_3$};

\path (3.2,-0.2) node [style=sergio]{$\xi'$};
\filldraw[red] (0,0) circle (2pt);
\path (0.1,0.15) node [style=sergio]{$v$};
\filldraw[red] (3,0) circle (2pt);
\path (3.1, 0.15) node [style=sergio]{$v$};

\filldraw[fill=gray!50, draw=gray, dotted]
(2.5,0.5-2.2) -- (3,0-2.2) -- (2.5,-0.5-2.2) -- cycle;
\filldraw[fill=gray!50, draw=gray, dotted]
(2.5,0.5-2.2) -- (3.5,0.5-2.2) -- (3,0-2.2) -- cycle;
\filldraw[fill=gray!50, draw=gray, dotted]
(3.5,0.5-2.2) -- (3,0-2.2) -- (3.5,-0.5-2.2) -- cycle;

\draw (2,0-2.2)--(3,0-2.2) node[currarrow,pos=0.5,sloped]{};
\draw (3,0-2.2)--(4,0-2.2) node[currarrow,pos=0.5,sloped]{};
\draw (3,0-2.2)--(3,-1-2.2) node[currarrow,pos=0.5,sloped]{};
\draw (3,1-2.2)--(3,0-2.2) node[currarrow,pos=0.5,sloped]{};

\path (1.7,-2.2) node [style=sergio]{$=$};
\path (2.5,0.2-2.2) node [style=sergio]{\small $x_2 \bar{k} h k$};
\path (3.5,0.2-2.2) node [style=sergio]{\small $\bar{k} \bar{h} k x_3$};
\path (3,0.6-2.1) node [style=sergio]{\small $y_5 \bar{k} h k$};
\path (3-0.2,-0.4-2.2) node [style=sergio]{\small $y_3$};

\filldraw[red] (3,-2.2) circle (2pt);

\end{tikzpicture}
\end{equation}

The new ribbon operator becomes

\begin{equation}
    \begin{aligned}
        F^{\bar{k} \bar{h} k, e}_{\xi'} F^{h,g}_{\xi} &= \sum_{k,l} F^{h,k}_{\xi_1} F^{\bar{k} h k, l}_{\xi_2'} F^{\overline{(kl)} h (kl), \overline{(kl)} g}_{\xi_3} \\
        &= F^{h,g}_{\xi''},
    \end{aligned}
\end{equation}
where $\xi''$ is the new ribbon below.
\begin{equation}
\begin{tikzpicture}[scale=1.4]
\tikzstyle{sergio}=[rectangle,draw=none]
\filldraw[fill=gray!50, draw=gray, dotted]
(-4,-1) -- (-2,-1) -- (-2.5,-1.5) -- (-4.5,-1.5) -- cycle;
\filldraw[fill=gray!50, draw=gray, dotted]
(-2,-1) -- (-1,-1) -- (-1.5,-1.5) -- cycle;
\filldraw[fill=gray!50, draw=gray, dotted]
(-1,-1) -- (0,-1) -- (-0.5,-1.5) -- (-1.5,-1.5) -- cycle;
\filldraw[fill=gray!50, draw=gray, dotted]
(-2.5,-1) -- (-2.5,-0.5) -- (-1.5,-0.5) -- (-1.5,-1) -- cycle;

\foreach \x in {-4,-3,...,-1} {
  \draw (\x,-1)--(\x+1,-1) node[currarrow,pos=0.5,sloped]{};
}

\draw (-2,0)--(-2,-1) node[currarrow,pos=0.5,sloped]{};

\foreach \x in {-4,-3,...,-1} {
  \draw (\x,-2)--(\x+1,-2);
}

\foreach \x in {-4,-3,...,0} {
  \draw (\x,-1)--(\x,-2) node[currarrow,pos=0.5,sloped]{};
}

\draw[draw=gray, dotted] (-2.5,-0.5)--(-2,-1);
\draw[draw=gray, dotted] (-1.5,-0.5)--(-2,-1);

\filldraw[red] (-2,-1) circle (2pt);
\path (-1.9,-0.8) node [style=sergio]{$v$};

\end{tikzpicture}
\end{equation}

By sequentially applying these local ribbon operators, one can deform a ribbon operator to arbitrary shapes. 

One could also consider the plaquette-flatness constraint,

\begin{equation}
\begin{tikzpicture}[scale=1.42]
\tikzstyle{sergio}=[rectangle,draw=none]

\filldraw[fill=gray!50, draw=gray, dotted]
(-0.5,0.5) -- (0,0) -- (-0.5,-0.5) -- cycle;
\filldraw[fill=gray!50, draw=gray, dotted]
(-0.5,0.5) -- (0.5,0.5) -- (0,0) -- cycle;
\filldraw[fill=gray!50, draw=gray, dotted]
(0.5,0.5) -- (0,0) -- (0.5,-0.5) -- cycle;
\filldraw[fill=gray!50, draw=gray, dotted]
(0.5,-0.5) -- (0,0) -- (-0.5,-0.5) -- cycle;

\draw (-0.5,0.5)--(0.5,0.5) node[currarrow,pos=0.5,sloped]{};
\draw (0.5,0.5)--(0.5,-0.5) node[currarrow,pos=0.5,sloped]{};
\draw (-0.5,0.5)--(-0.5,-0.5) node[currarrow,pos=0.5,sloped]{};
\draw (-0.5,-0.5)--(0.5,-0.5) node[currarrow,pos=0.5,sloped]{};

\path (-1.2,0) node [style=sergio]{$F^{h, e}_{\xi'}$};
\path (-0.7,0) node [style=sergio]{$y_5$};
\path (0.7,0) node [style=sergio]{$y_6$};
\path (0,0.7) node [style=sergio]{$x_5$};
\path (0,-0.7) node [style=sergio]{$x_3$};
\path (1.5,0) node [style=sergio]{$= \delta_{x_5 y_6 \bar{x}_3 \bar{y}_5, e}$};
\path (0.2,-0.2) node [style=sergio]{$\xi'$};

\filldraw[fill=gray!50, draw=gray, dotted]
(-0.5+3,0.5) -- (0+3,0) -- (-0.5+3,-0.5) -- cycle;
\filldraw[fill=gray!50, draw=gray, dotted]
(-0.5+3,0.5) -- (0.5+3,0.5) -- (0+3,0) -- cycle;
\filldraw[fill=gray!50, draw=gray, dotted]
(0.5+3,0.5) -- (0+3,0) -- (0.5+3,-0.5) -- cycle;
\filldraw[fill=gray!50, draw=gray, dotted]
(0.5+3,-0.5) -- (0+3,0) -- (-0.5+3,-0.5) -- cycle;

\draw (-0.5+3,0.5)--(0.5+3,0.5) node[currarrow,pos=0.5,sloped]{};
\draw (0.5+3,0.5)--(0.5+3,-0.5) node[currarrow,pos=0.5,sloped]{};
\draw (-0.5+3,0.5)--(-0.5+3,-0.5) node[currarrow,pos=0.5,sloped]{};
\draw (-0.5+3,-0.5)--(0.5+3,-0.5) node[currarrow,pos=0.5,sloped]{};

\path (-0.7+3,0) node [style=sergio]{$y_5$};
\path (0.7+3,0) node [style=sergio]{$y_6$};
\path (0+3,0.7) node [style=sergio]{$x_5$};
\path (0+3,-0.7) node [style=sergio]{$x_3$};

\path (0.2+3,-0.2) node [style=sergio]{$\xi'$};

\filldraw[red] (-0.5,-0.5) circle (2pt);
\path (0.1-0.5,0.15-0.5) node [style=sergio]{$v$};
\filldraw[red] (3-0.5,-0.5) circle (2pt);
\path (3.1-0.5, 0.15-0.5) node [style=sergio]{$v$};

\end{tikzpicture}
\end{equation}

Near the vertex $v$, we are free to insert this term since the plaquette is flat. We have

\begin{equation}
    \begin{aligned}
        F^{\bar{k} h k,e}_{\xi'} F^{h,g}_{\xi} = \sum_{k,l} F^{h,k}_{\xi_1} F^{\bar{k} h k, e}_{\xi'} F^{\bar{k} h k, l}_{\xi_2} F^{\overline{(kl)} h (kl), \overline{(kl)} g}_{\xi_3}
    \end{aligned}
\end{equation}

Near the vertex $v$, the operator acts as follows.

\begin{equation}
\begin{tikzpicture}[scale=1.42]
\tikzstyle{sergio}=[rectangle,draw=none]

\filldraw[fill=gray!50, draw=gray, dotted]
(-0.5,0.5) -- (0,0) -- (-0.5,-0.5) -- cycle;
\filldraw[fill=gray!50, draw=gray, dotted]
(-0.5,0.5) -- (0.5,0.5) -- (0,0) -- cycle;
\filldraw[fill=gray!50, draw=gray, dotted]
(0.5,0.5) -- (0,0) -- (0.5,-0.5) -- cycle;
\filldraw[fill=gray!50, draw=gray, dotted]
(0.5,-0.5) -- (0,0) -- (-0.5,-0.5) -- cycle;

\draw (-0.5,0.5)--(0.5,0.5) node[currarrow,pos=0.5,sloped]{};
\draw (0.5,0.5)--(0.5,-0.5) node[currarrow,pos=0.5,sloped]{};
\draw (-0.5,0.5)--(-0.5,-0.5) node[currarrow,pos=0.5,sloped]{};
\draw (-0.5,-0.5)--(0.5,-0.5) node[currarrow,pos=0.5,sloped]{};

\path (-1.5,0) node [style=sergio]{$F^{\bar{k} h k,e}_{\xi'} F^{\bar{k} h k,l}_{\xi_2}$};
\path (-0.7,0) node [style=sergio]{$y_5$};
\path (0.7,0) node [style=sergio]{$y_6$};
\path (0,0.7) node [style=sergio]{$x_5$};
\path (0,-0.7) node [style=sergio]{$x_3$};
\path (1.5,0) node [style=sergio]{$= \delta_{\bar{y}_5 x_5 y_6, l}$};
\path (0.2,-0.2) node [style=sergio]{$\xi'$};

\filldraw[fill=gray!50, draw=gray, dotted]
(-0.5+3,0.5) -- (0+3,0) -- (-0.5+3,-0.5) -- cycle;
\filldraw[fill=gray!50, draw=gray, dotted]
(-0.5+3,0.5) -- (0.5+3,0.5) -- (0+3,0) -- cycle;
\filldraw[fill=gray!50, draw=gray, dotted]
(0.5+3,0.5) -- (0+3,0) -- (0.5+3,-0.5) -- cycle;
\filldraw[fill=gray!50, draw=gray, dotted]
(0.5+3,-0.5) -- (0+3,0) -- (-0.5+3,-0.5) -- cycle;

\draw (-0.5+3,0.5)--(0.5+3,0.5) node[currarrow,pos=0.5,sloped]{};
\draw (0.5+3,0.5)--(0.5+3,-0.5) node[currarrow,pos=0.5,sloped]{};
\draw (-0.5+3,0.5)--(-0.5+3,-0.5) node[currarrow,pos=0.5,sloped]{};
\draw (-0.5+3,-0.5)--(0.5+3,-0.5) node[currarrow,pos=0.5,sloped]{};

\path (-0.7+3,0) node [style=sergio]{$y_5$};
\path (0.7+3,0) node [style=sergio]{$y_6$};
\path (0+3,0.7) node [style=sergio]{$x_5$};
\path (0+3,-0.7) node [style=sergio]{$l$};

\path (0.2+3,-0.2) node [style=sergio]{$\xi'$};

\filldraw[red] (-0.5,-0.5) circle (2pt);
\path (0.1-0.5,0.15-0.5) node [style=sergio]{$v$};
\filldraw[red] (3-0.5,-0.5) circle (2pt);
\path (3.1-0.5, 0.15-0.5) node [style=sergio]{$v$};

\end{tikzpicture}
\end{equation}

The original ribbon operator becomes 

\begin{equation}
    \begin{aligned}
        F^{\bar{k} h k,e}_{\xi'} F^{h,g}_{\xi} &= \sum_{k,l} F^{h,k}_{\xi_1} F^{\bar{k} h k, l}_{\xi_2'} F^{\overline{(kl)} h (kl), \overline{(kl)} g}_{\xi_3} \\
        &= F^{h,g}_{\xi''},
    \end{aligned}
\end{equation}
where $\xi''$ is the new ribbon as illustrated below.

\begin{equation}
\begin{tikzpicture}[scale=1.4]
\tikzstyle{sergio}=[rectangle,draw=none]
\filldraw[fill=gray!50, draw=gray, dotted]
(-4,-1) -- (-2,-1) -- (-2.5,-1.5) -- (-4.5,-1.5) -- cycle;
\filldraw[fill=gray!50, draw=gray, dotted]
(-2,-1) -- (-1.5,-1) -- (-2,-1.5) -- (-2.5,-1.5) -- cycle;
\filldraw[fill=gray!50, draw=gray, dotted]
(-1.5,-1) -- (0,-1) -- (-0.5,-1.5) -- (-1,-1.5) -- cycle;
\filldraw[fill=gray!50, draw=gray, dotted]
(-2,-1) -- (-2,0) -- (-1,0) -- (-1,-1) -- cycle;

\foreach \x in {-4,-3,...,-1} {
  \draw (\x,-1)--(\x+1,-1) node[currarrow,pos=0.5,sloped]{};
}
\foreach \x in {-4,-3,...,-1} {
  \draw (\x,-2)--(\x+1,-2);
}

\foreach \x in {-4,-3,...,0} {
  \draw (\x,-1)--(\x,-2) node[currarrow,pos=0.5,sloped]{};
}

\draw (-2,0)--(-2,-1) node[currarrow,pos=0.5,sloped]{};
\draw (-1,0)--(-1,-1) node[currarrow,pos=0.5,sloped]{};
\draw (-2,0)--(-1,0) node[currarrow,pos=0.5,sloped]{};

\draw[draw=gray, dotted] (-2,0)--(-1,-1);
\draw[draw=gray, dotted] (-1,0)--(-2,-1);

\filldraw[red] (-2,-1) circle (2pt);
\path (-1.9,-0.8) node [style=sergio]{$v$};

\end{tikzpicture}
\end{equation}

\section{\texorpdfstring{$S$}{S}-matrix of quantum double model}

In this section, we review the modular properties of the quantum double $D(G)$ and present explicit formulas for the $S$-matrix. We then show that for any nontrivial magnetic flux (electric charge), there exists an electric charge (magnetic flux) that braids nontrivially with it, producing a nontrivial braiding factor.

Anyons in the non-abelian quantum double $D(G)$ are given by the pair $(C(g),\Gamma)$, where the flux $C(g)=\{hgh^{-1}|h\in G\}$ is a conjugacy class, and the charge $\Gamma$ is an irreducible representation of the centralizer $Z_{g}=\{h|hg=gh\}$.

The $S$-matrix is given by the following formula~\cite{beigi2010quantum-2ab},
\begin{equation}
    \begin{aligned}
        S_{(C(g), \Gamma)(C(g'), \Gamma')} = \frac{\sum_{h:hg'\bar{h} \in Z_g} \chi_{\Gamma} (h \bar{g}' \bar{h}) \chi_{\Gamma'} (\bar{h} \bar{g} h)}{|Z_g| |Z_{g'}|}.
    \end{aligned}
\end{equation}

For a purely electric charge $([e], \Gamma)$ and a purely magnetic flux $(C(g), 1)$, the $S$-matrix is given as follows,
\begin{equation}
    \begin{aligned}
        S_{(C(e), \Gamma)(C(g), 1)} = \frac{\sum_{h \in G} \chi_{\Gamma}(h \bar{g} \bar{h})}{|G||Z_g|} = \frac{|C(g)|}{|G|}  \chi_{\Gamma}(\bar{g}),
    \end{aligned}
\end{equation}
where $|C(g)|$ is the order of the conjugacy class $C(g)$.

\section{Proof of Theorem 1}

\begin{proof}
Firstly, we prove that all density matrices in $\mathcal{Q}\{\rho_{\mathrm{cl}}\}$ with different eigenvalues of non-contractible strong 1-form symmetry defined in Eq.~\eqref{Eq: strong sym} are identical on any simply connected subregion $A$. 

For all ground states $\ket{\psi_{\mathrm{QD}}}$ of $\mathcal{D}(G)$, the reduced density matrix on any simply connected subregion $A$ is the same,
\begin{align} \tr_{\bar{A}}(\ket{\psi_{\mathrm{QD}}}\bra{\psi_{\mathrm{QD}}})=\rho_A, \end{align}
which is a key consequence of the spontaneous symmetry breaking of 1-form symmetries (’t Hooft anomaly).

Then consider the decoherence channel defined in Eq. \eqref{Eq: decohered QD}. Since the Kraus operators act independently on different sites, the channel can be decomposed into parts on $A$ and $\bar{A}$, namely $\N=\N_A\circ\N_{\bar{A}}$. Then consider a decohered density matrix $\cl=\N[\ket{\psi_{\mathrm{QD}}}\bra{\psi_{\mathrm{QD}}}]$. We have the following relation,
\begin{align}
\cl\big|_A=\tr_{\bar{A}}(\cl)=\N_A[\rho_A], 
\end{align}
which is also indistinguishable in the region $A$. 

Explicitly, we have
\begin{equation}
    \begin{aligned}
        \cl = \mathcal{N}[\rho_{QD}] &=\sum_{\{Z_{\Gamma}\}_A,\{Z_{\Gamma'}\}_{\bar{A}},\Gamma, \Gamma'\in \mathrm{Rep}(G)} \frac{d_{\Gamma}^{|\{Z_{\Gamma}\}_A|} d_{\Gamma'}^{ |\{Z_{\Gamma'}\}_{\bar{A}}|}}{|G|^{|A| + |\bar{A}|}} \{Z_{\Gamma, \alpha, \alpha'}\}_{A} \{Z_{\Gamma', \beta, \beta'}\}_{\bar{A}} \rho_{QD} \{Z_{\Gamma, \alpha, \alpha'}^{\dagger}\}_{A} \{Z_{\Gamma', \beta, \beta'}^{\dagger}\}_{\bar{A}},
    \end{aligned}
\end{equation}
where $\{Z_{\Gamma}\}_A$ represents all the possible $Z_{\Gamma}$ configurations in the region $A$, $Z_{\Gamma, \alpha, \alpha'}$ represents the $(\alpha, \alpha')$ component of the internal states of $Z_{\Gamma}$, $|\{Z_{\Gamma}\}_A|$ denotes the number of edges in region $A$ that $Z_{\Gamma}$ is applied on, and $d_{\Gamma}^{|\{Z_{\Gamma}\}_A|} := \prod_{i} d_{\Gamma_{i}}$, where $Z_{\Gamma_i} \in \{Z_{\Gamma}\}_A$. For the simplicity of notations, we sum over the internal degrees of freedom $\alpha, \alpha', \beta, \beta'$ automatically. Then we have
{\footnotesize
    \begin{equation*}
        \begin{aligned}
            \tr_{\bar{A}} (\cl) &= \tr_{\bar{A}} \left(\sum_{\Gamma, \Gamma' \in \mathrm{Rep}(G)}\frac{1}{|G|^{|A| + |\bar{A}|}} \sum_{\{Z_{\Gamma}\}_A,\{Z_{\Gamma'}\}_{\bar{A}}} d_{\Gamma}^{|\{Z_{\Gamma}\}_A|} d_{\Gamma'}^{ |\{Z_{\Gamma'}\}_{\bar{A}}|} \{Z_{\Gamma, \alpha, \alpha'}\}_{A} \{Z_{\Gamma', \beta, \beta'}\}_{\bar{A}} \rho_{QD} \{Z_{\Gamma, \alpha, \alpha'}^{\dagger}\}_{A} \{Z_{\Gamma', \beta, \beta'}^{\dagger}\}_{\bar{A}}\right) \\
            &= \sum_{\{Z_{\Gamma}\}_A, \Gamma \in \mathrm{Rep}(G)} \frac{d_{\Gamma}^{|\{Z_{\Gamma}\}_{A}|}}{|G|^{|A|}} \{Z_{\Gamma, \alpha, \alpha'}\}_{A} \tr_{\bar{A}} \left(\sum_{\Gamma' \in \mathrm{Rep}(G)}\frac{1}{|G|^{|\bar{A}|}} \sum_{\{Z_{\Gamma'}\}_{\bar{A}}} d_{\Gamma'}^{|\{Z_{\Gamma'}\}_{\bar{A}}|}  \{Z_{\Gamma', \beta, \beta'}\}_{\bar{A}} \rho_{QD}  \{Z_{\Gamma', \beta, \beta'}^{\dagger}\}_{\bar{A}}\right) \{Z_{\Gamma, \alpha, \alpha'}^{\dagger}\}_{A} \\
            &= \sum_{\{Z_{\Gamma}\}_A, \Gamma \in \mathrm{Rep}(G)} \frac{d_{\Gamma}^{|\{Z_{\Gamma}\}_{A}|}}{|G|^{|A|}} \{Z_{\Gamma, \alpha, \alpha'}\}_{A} \tr_{\bar{A}} \left(\rho_{QD}\right) \{Z_{\Gamma, \alpha, \alpha'}^{\dagger}\}_{A} \\
            &= \mathcal{N}_{A} [\rho_A].
        \end{aligned}
    \end{equation*}}
From the first to the second line, we use the fact that Kraus operators supported on region $A$ are unaffected by the partial trace. From the second to the third line, we use the following facts,
\begin{equation}
    \begin{aligned}
        \tr \left(Z_{\Gamma, \alpha, \alpha'} \rho Z_{\Gamma, \alpha, \alpha'}^{\dagger} \right) = d_{\Gamma} \tr \left(\rho\right),
    \end{aligned}
\end{equation}
and
\begin{equation}
    \begin{aligned}
        \sum_{\Gamma \in \mathrm{Rep}(G)} d_{\Gamma}^2 = |G|.
    \end{aligned}
\end{equation}

Consider starting from a non-trivial ground state of the quantum double model. The non-contractible loops can always be deformed to lie entirely within the region $\bar{A}$~\cite{Kitaev_2003,Bombin2007}. After tracing out $\bar{A}$, the resulting reduced density matrix must be identical, namely
\begin{align}
\cl^A = \tr_{\bar{A}}(\cl^1) = \tr_{\bar{A}}(\cl^2),
\end{align}
where $\cl^1$ and $\cl^2$ are the density matrices obtained by decohering different ground states of the quantum double.

Then we show that all density matrices in $\mathcal{Q}\{\cl\}$ have finite Markov length. We follow the proof in Ref.~\cite{sang2025mixed}. In the following expressions, we display only the local commuting-projector factors. The factors specifying the global holonomy sector are implicit and can be supported entirely in $C$, as discussed below. Before the decoherence, the density matrix can be written in a stabilizer form
\begin{equation}
    \begin{aligned}
        \rho_{\mathrm{QD}} \propto \prod_{p} B_p \prod_{v} A_v.
    \end{aligned}
\end{equation}
After the decoherence, since the Kraus operators don't commute with $A_v$ terms, while they commute with all the $B_p$ terms, the decohered density matrix is thus given by
\begin{equation}
    \begin{aligned}
        \cl \propto \prod_{p} B_p.
    \end{aligned}
\end{equation}
Each $B_p$ is a projector, and different $B_p$ projectors commute with each other. The generating set of the commuting projector state can be denoted as follows,
\begin{equation}
    \begin{aligned}
        \mathcal{G} (\cl) = \{B_p\}.
    \end{aligned}
\end{equation}

According to the discussion in Ref.~\cite{HAMMA200522,brown2012quantum,sang2025mixed}, the conditional mutual information $I(A:C|B)$ is given as follows,
\begin{equation}
    \begin{aligned}
        \quad \ I_{\cl}(A:C|B) = \min_{\mathcal{G}(\rho)} |\{B_p \in \mathcal{G}(\rho), \text{support of $B_p$ intersects $A$ and $C$}\}|,
    \end{aligned}
\end{equation}
where the minimization is taken over all equivalent choices of $\mathcal{G}(\rho)$. 

Therefore, for a sufficiently thick buffer region $B$, we can always have the following result,
\begin{equation}
    \begin{aligned}
        I_{\cl}(A:C|B) = 0,
    \end{aligned}
\end{equation}
which corresponds to the Markov length $\xi(\cl) = 0$. In particular, we claim that all ribbon operators $F_\xi^\Gamma$ defined on non-contractible loops do not contribute to CMI. Because only topology matters for the ribbon operators, we can always deform the non-contractible loops such that all corresponding ribbon operators are supported solely at $C$.
\end{proof}

\section{More properties of Kitaev's quantum double model}
Before we discuss the dimension of the locally indistinguishable set $\mathcal{Q}\{\cl\}$, we first review the dimension of the ground-state subspace (codespace) of Kitaev's quantum double model in isolated systems \cite{Cui_2020}. 

\begin{thm}
\label{Thm: codespace}
Consider the (untwisted) Kitaev’s $\mathcal{D}(G)$ quantum double defined on a cellulation $\Sigma=(V,E,F)$ of a closed 2-manifold. The dimension of the codespace $\mathcal{C}[\mathcal{D}(G)]$ is equal to the number of $G$-orbits of $\mathrm{Hom}(\pi_1(\Sigma),G)$ under conjugation:
\begin{align}
\mathrm{dim}\ \mathcal{C}[\mathcal{D}(G)]=|\mathrm{Hom}(\pi_1(\Sigma),G)/\mathrm{Ad}_G|, \label{eq: dim_codespace}
\end{align} 
where 
\begin{align}
    \mathrm{Ad} : G \to \mathrm{Aut} (G),\quad \mathrm{Ad}_g(x):= g x \bar{g},
\end{align}
and $\pi_1(\Sigma)$ is the fundamental group of $\Sigma$.
\end{thm}

\begin{proof}
Consider the group basis $\ket{g} = \bigotimes_{e \in E} \ket{g_e}$, where each edge $e \in E$ is assigned a group element $g_e \in G$. Let $\gamma$ be an oriented path in $\Sigma$, represented as an ordered sequence of connected edges. We define  
\begin{equation}
    g_\gamma := \prod_{e \in \gamma} g_e^{\sigma(e,\gamma)},
\end{equation}
where $\sigma(e,\gamma) = +1$ if the orientation of $e$ agrees with that of $\gamma$, and $\sigma(e,\gamma) = -1$ otherwise.

For the ground state of the quantum double model, the condition $B_p |\Psi_{\mathrm{QD}}\rangle = |\Psi_{\mathrm{QD}}\rangle$ requires $g_{\partial p} = e$, where $e \in G$ is the identity element and $\partial p$ denotes the boundary of plaquette $p$ oriented counterclockwise. The $B_p$ stabilizers thus define the subspace
\begin{align}
S &= \left\{ \ket{g} \,\big|\, g_{\partial p} = e,~ \forall p \in F \right\} \nonumber \\
  &= \left\{ \ket{g} \,\big|\, g_\gamma = e,~ \forall \text{ contractible loops } \gamma \right\},
\end{align}
And we have 
\begin{equation}
    B_p |g \rangle = |g\rangle, \ \forall |g\rangle \in S.
\end{equation}

Next, consider the gauge transformation operator $A_v^h$ for $\forall h \in G$.  
For any state $\ket{g} \in S$, it is straightforward to verify that  
\begin{align}
    A_v^h \ket{g} = \ket{g'} \in S,
\end{align}
since, by definition in Eq.~\eqref{Eq: AvBp}, the action of $A_v^h$ does not change the group element associated with any \emph{contractible loop}.  

Therefore, two states $\ket{g}, \ket{g'} \in S$ are said to be \textit{gauge equivalent} if they can be related by gauge transformations acting on a finite set of vertices. We denote this equivalence by $\ket{g} \sim \ket{g'}$, and define the corresponding equivalence classes of $S$ as  
\begin{align}
    \ket{[g]} \propto \sum_{\ket{g'} \sim \ket{g}} \ket{g'} \in [S].
\end{align}
It follows directly that $\ket{[g]}$ is invariant under $A_v^h$ for all $h \in G$. Hence, the $A_v$ terms also serve as stabilizers of these states, and the set $\left\{ \ket{[g]} \,\big|\, [g] \in [S] \right\}$ forms a basis of the codespace $\mathcal{C}[\mathcal{D}(G)]$.

Choose an arbitrary vertex $v_0$ as a base point of the square lattice defined upon the topological space $\Sigma$ and a maximal spanning tree $T$ (a maximal subgraph of $\Sigma$ that does not contain any loops, with exactly $m=|V|-1$ edges) containing $v_0$. 

Define a map:
\begin{align}
\Phi:~S\longrightarrow\mathrm{Hom}(\pi_1(\Sigma),G).
\end{align}
Let $\gamma$ be an arbitrary closed loop starting and ending at $v_0$, which may be either contractible or non-contractible.  
For any $\ket{g} \in S$, define
\begin{equation}
    \Phi(\ket{g})([\gamma]) = g_\gamma,
\end{equation}
which maps a closed path $\gamma$ to its \textit{holonomy}, defined as the product of the group elements along the loop.  

By definition, the holonomy of any contractible loop $\gamma_0$ must be trivial: $g_{\gamma_0} = e$. This implies that $\Phi(\ket{g})$ depends only on the homotopy class of $\gamma \in \pi_1(\Sigma, v_0)$. Therefore, the map $\Phi(\ket{g})$ is well-defined as a homomorphism from $\pi_1(\Sigma, v_0)$ to $G$.

Next, we show that the set of gauge-equivalence classes $[S]$ is in one-to-one correspondence with the orbits of $\mathrm{Hom}(\pi_1(\Sigma),G)$ under the conjugation action of $G$.  

Fix $\phi \in \mathrm{Hom}(\pi_1(\Sigma),G)$. To construct a preimage $\ket{g} \in S$ of $\phi$, consider any edge $e \notin T$ with endpoints $\partial_0 e$ and $\partial_1 e$. By construction of the maximal spanning tree $T$, there exist unique paths $\gamma_0, \gamma_1 \subset T$ connecting $v_0$ to $\partial_0 e$ and $\partial_1 e$, respectively. Define $\bar{\gamma}_i$ as the reversed path of $\gamma_i$. The loop $\gamma = \gamma_0  e  \bar{\gamma}_1$ is a closed path based at $v_0$. The holonomy condition requires that
\begin{align}
    g_{\gamma_0} g_e g_{\bar{\gamma}_1} = \phi(\gamma).
\end{align}
Thus, for each edge $e \notin T$, there is a unique assignment of $g_e \in G$ consistent with $\phi$, while the group elements on the edges of $T$ remain arbitrary. Since $T$ contains $m$ edges, there are $|G|^m$ distinct preimages $\ket{g}$ associated with a fixed $\phi$. Hence, $\Phi$ is surjective and $|G|^m$-to-$1$.  

Next, consider gauge transformations. If $\ket{g}$ and $\ket{g'}$ differ by gauge transformations on vertices other than $v_0$, the holonomies along closed loops remain unchanged, so $\Phi(\ket{g}) = \Phi(\ket{g'})$. On the other hand, performing a gauge transformation $A_{v_0}^h$ at the root vertex $v_0$ modifies every holonomy by conjugation:
\begin{align}
    \Phi\!\left(A_{v_0}(h)\ket{g}\right)(\gamma) = h \, \Phi(\ket{g})(\gamma) \, h^{-1},
\end{align}
since each loop $\gamma$ both begins and ends at $v_0$.  

Therefore, gauge-equivalence classes $[S]$ correspond precisely to orbits of $\mathrm{Hom}(\pi_1(\Sigma),G)$ under the conjugation action of $G$. We thus have
\begin{align}
    \mathrm{dim}\ C[\mathcal{D}(G)] = |[S]| = |\mathrm{Hom}(\pi_1(\Sigma),G) / \mathrm{Ad}_G|.
\end{align}
\end{proof}

As an illustrative example, let us explicitly compute the dimension of the codespace of the quantum double $\mathcal{D}(S_3)$ on the torus $\Sigma \simeq T^2$. In this case,  
\begin{align}
    \mathrm{Hom}(\pi_1(T^2), S_3) 
    = \left\{ (g,h) \in S_3 \times S_3 \,\big|\, gh = hg \right\}.
\end{align}
Equivalently, this set consists of pairs where $h$ lies in the centralizer of $g$, $Z_g = \{h \in S_3 \mid hg = gh \}$. Counting such commuting pairs gives a total of $18$ elements.

Then we consider the $G$-orbits on $\mathrm{Hom}(\pi_1(T^2), S_3)$ under conjugation, which leads to the following pairs up to global conjugation,
\begin{equation}
    \begin{aligned}
    &(e, e),\ (e, c),\ (e, t),\ (c, c), \\
    &(c, c^2),\ (c, e),\ (t, e),\ (t, t),
    \end{aligned}
\end{equation}
where $t, c \in S_3$, $c^3 = t^2 = e$, and $t c t = c^2$. Therefore, the dimension of the codespace $\mathcal{C}[\mathcal{D}(S_3)]$ is:
\begin{align}
\left|\mathrm{Hom}(\pi_1(T^2),S_3)/\mathrm{Ad}_{S_3}\right|=8.
\end{align}
Motivated by this example, we state the following corollary, which relates the codespace dimension to the different types of anyons.
\begin{corollary}
For a lattice defined on a surface $\Sigma$ that is topologically equivalent to a torus $\Sigma \simeq T^2$, the dimension of the codespace of Kitaev’s $\mathcal{D}(G)$ quantum double is equal to the number of distinct anyon types.
\begin{align}
\left|\mathrm{Hom}(\pi_1(T^2),G)/\mathrm{Ad}_G\right|=\#\text{ types of anyons.}
\end{align}
\end{corollary}

\begin{proof}
From the previous discussion, the holonomies on the torus are labeled by  
\begin{align}
    \mathrm{Hom}(\pi_1(T^2),G) = \{(g,h) \in G \times G \,\big|\, hg = gh \}.
\end{align}
Consider two pairs $(g,h_1)$ and $(g,h_2)$ that satisfy the above condition, with $h_1 \neq h_2$. These two pairs are inequivalent if and only if $h_1$ and $h_2$ belong to different conjugacy classes of the centralizer group $Z_g$.  

Therefore, the holonomies can be labeled by
\begin{equation}
    (C_G, C_{Z_g}), \quad g \in C_G,
\end{equation}
where $C_G$ denotes the conjugacy class of $G$, and $C_{Z_g}$ denotes a conjugacy class in the centralizer $Z_g$, $g \in C_G$.  

In particular, it is well known that for any finite group $H$, the number of distinct irreducible representations equals the number of its conjugacy classes. Hence, we may relabel the elements of the codespace $\mathcal{C}[\mathcal{D}(G)]$ as
\begin{align}
    (C,R) \in \mathrm{Hom}(\pi_1(T^2),G) / \mathrm{Ad}_G,
\end{align}
where $C$ labels a conjugacy class of $G$, and $R$ labels an irrep of the corresponding centralizer $Z_C$.

On the other hand, in Kitaev’s quantum double model, a pair of anyons at $s_0$ and $s_1$ can be created by a ribbon operator $F_\xi^{C,R}$ along a path $\xi$ with $\partial \xi = \{s_0,s_1\}$. Distinct ribbon operators, corresponding to different anyon types, are labeled by $(C,R)$, where $C$ denotes a conjugacy class of $G$ and $R$ an irrep of the centralizer $Z_C$.  

Therefore, we conclude that on torus, the codespace $\mathcal{C}[\mathcal{D}(G)]$ is in one-to-one correspondence with the different anyon types.
\end{proof}

Now we turn to the mixed states $\rho_{\mathrm{cl}}$ obtained by the decoherence Eq.~\eqref{Eq: decohered QD}, with the corresponding 1-form strong and weak symmetries defined as Eqs. \eqref{Eq: strong sym} and \eqref{Eq: weak sym}.  

\section{Proof of Theorem 2}

\begin{proof}
From the previous discussion, we know that in the pure-state quantum double model, the degenerate ground states are labeled by the holonomies of distinct non-contractible loops, as given in Eq.~\eqref{eq: dim_codespace}.

First, consider decohering the ground state without non-contractible loops, $|e, e\rangle$. This state can be expressed as a superposition of all contractible loops labeled by $F^{g,e}_{\xi}$, where $\xi$ is any contractible ribbon. The corresponding density matrix takes the form
\begin{equation}
\rho_{\mathrm{QD}}^{e,e} = |e, e\rangle \langle e, e| \propto \sum_{\mathbf{l}, \mathbf{l}'} |\mathbf{l}\rangle \langle \mathbf{l}'|,
\end{equation}
where each $|\mathbf{l}\rangle$ denotes a contractible loop configuration. Summing over all contractible loop configurations yields the density matrix of the trivial ground state.

When $\mathbf{l} \neq \mathbf{l}'$, there must exist a site $s$ such that $|g_s\rangle \neq |g_s'\rangle$, where $|g_s\rangle$ and $|g_s'\rangle$ denote the states at site $s$ for $|\mathbf{l}\rangle$ and $|\mathbf{l}'\rangle$, respectively. We then have
\begin{equation}
    \begin{aligned}
        &\quad \ \sum_{\Gamma} d_{\Gamma} Z_{\Gamma,s, \alpha, \alpha'} | \mathbf{l} \rangle \langle \mathbf{l}'| Z_{\Gamma, s, \alpha, \alpha'}^{\dagger} \\
        &= \sum_{\Gamma} d_{\Gamma} \chi_{\Gamma}(g_s \bar{g}_s') | \mathbf{l} \rangle \langle \mathbf{l}'| \\
        &= \delta_{g_s \bar{g}_s',e} |G| | \mathbf{l} \rangle \langle \mathbf{l}'|.
    \end{aligned}
\end{equation}
From the second to the third line, we apply the great orthogonality theorem. When $g_s \neq g_s'$, the expression above vanishes by this theorem. Consequently, the corresponding off-diagonal contribution also vanishes after applying the full channel, since $\mathcal{N} = \mathcal{N}_{\bar{s}} \circ \mathcal{N}_{s}$, with $\bar{s}$ denoting the complement of $s$.

Therefore, after decoherence, the trivial ground state $\rho_{\mathrm{QD}}^{e,e}$ reduces to a sum over classical loop configurations.
\begin{equation}
    \begin{aligned}
        \mathcal{N} [\rho_{\mathrm{QD}}^{e,e}] \propto \sum_{\mathbf{l}} | \mathbf{l} \rangle \langle \mathbf{l}|, \label{eq: classical_loop}
    \end{aligned}
\end{equation}
for all contractable loops $\mathbf{l}$.

In the general case, we consider a quantum double ground state of the form
\begin{equation}
|\Psi_{\mathrm{QD}}\rangle \propto \sum_{i=1}^{|\mathcal{C}[\mathcal{D}(G)]|} c_i  |g_i, h_i\rangle,
\end{equation}
where $g_i, h_i \in G$, and $|g_i, h_i\rangle$ denotes a distinct ground state for each $i$ satisfying Eq.~\eqref{eq: dim_codespace}, with $c_i$ being arbitrary coefficients. The extremal points are labeled by simultaneous conjugacy classes of the holonomy pairs $(g,h)$, and here we use one representative from each class.

By definition, the state above can be expressed as,
\begin{equation}
    \begin{aligned}
        |\Psi_{\mathrm{QD}}\rangle \propto \sum_{i=1}^{|\mathcal{C}[\mathcal{D}(G)]|} c_i F^{g_i, h_i}_{\xi_x} F^{h_i, e}_{\xi_y} |e, e\rangle,
    \end{aligned}
\end{equation}
where $\xi_x$ and $\xi_y$ denote two distinct non-contractable ribbons. We have
\begin{equation}
\begin{aligned}
\rho_{\mathrm{QD}} &\propto \sum_{i,i'}c_i c_{i'}^* |g_i,h_i\rangle\langle g_{i'},h_{i'}| \\
&\propto \sum_{i,i'}c_i c_{i'}^* F_{\xi_x}^{g_i,h_i}F_{\xi_y}^{h_i,e} |e,e\rangle\langle e,e| F_{\xi_y}^{\bar h_{i'},e}F_{\xi_x}^{\bar g_{i'},h_{i'}}.
\end{aligned}
\end{equation}
Consider two sites $s_1 \in \xi_x$ and $s_2 \in \xi_y$. We have
\begin{equation}
    \begin{aligned}
       \mathcal{N}_{s_1}[\rho_{\mathrm{QD}}] &\propto \sum_{\Gamma} d_{\Gamma} Z_{\Gamma, \alpha, \alpha'} \rho_{QD} Z_{\Gamma, \alpha, \alpha'}^{\dagger} \\
       &\propto \sum_{\Gamma} d_{\Gamma} \chi_{\Gamma}(g_i \bar{g}_{i'}) \sum_{\mathbf{l}}|\mathbf{l} \rangle \langle \mathbf{l} |.
    \end{aligned}
\end{equation}
From the first line to the second line we use the commutation property discussed in Eq.~\eqref{eq:conmmutation}, as well as the previous result Eq.~\eqref{eq: classical_loop}. By the great orthogonality theorem, the expression above vanishes when $g_i \neq g_{i'}$, which in turn causes the corresponding off-diagonal contribution to vanish under the full channel since $\mathcal{N} = \mathcal{N}_{s_1} \circ \mathcal{N}_{\bar{s}_1}$. Hence, the non-vanishing elements of the decohered density matrix must satisfy $g_i = g_{i'}$. A similar argument applies to the other label $h$. Therefore, the decohered density matrix can, in general, be written in the following form:
\begin{equation}
    \begin{aligned}
        \mathcal{N}[\rho_{\mathrm{QD}}] \propto \sum_{i, \mathbf{l}}^{|C[\mathcal{D}(G)]|} p_i |\mathbf{l}^{g_i, h_i} \rangle \langle \mathbf{l}^{g_i, h_i} |,
    \end{aligned}
\end{equation}
where $g_i, h_i \in G$ satisfy Eq.~\eqref{eq: dim_codespace}, and $|\mathbf{l}^{g_i, h_i} \rangle$ denotes a loop configuration with non-contractable loop labeled by $(g_i, h_i)$. The density matrix $\rho_i \propto \sum_{\mathbf{l}} |\mathbf{l}^{g_i, h_i} \rangle \langle \mathbf{l}^{g_i, h_i} |$ is an extremal point of the information convex set and satisfies
\begin{equation}
\begin{aligned}
\operatorname{Tr}(\rho_i\rho_{i'})=0, \qquad i\neq i'.
\end{aligned}
\end{equation}

Therefore, it is straightforward to see that
\begin{equation}
\mathrm{Ext}(\mathcal{Q}\{\rho_{\mathrm{cl}}\}) \cong \mathrm{Hom}(\pi_1(\Sigma),G)/\mathrm{Ad}_G,
\end{equation}
\end{proof}

\section{Example: Decohering \texorpdfstring{$S_3$}{S3} quantum double}

The group of permutations on a set of three elements $S_3$ is isomorphic to the dihedral group $D_3$, the symmetry of an equilateral triangle. There are two generators $c$ and $t$ satisfying the following relations
\begin{align}
    c^3 = t^2 = e,\quad tct = c^2,
\end{align}
where $e$ is the identity element. According to the previous discussion, the ground states of the $S_3$ quantum double on a torus are labeled as follows,
\begin{equation}
    \begin{aligned}
        &|e, e\rangle, \ |e, c\rangle,\ |e, t\rangle,\ |c, c\rangle, \\
        &|c, c^2 \rangle,\ |c, e\rangle,\ |t, e\rangle,\ |t, t\rangle.
    \end{aligned}
\end{equation}
The ground state degeneracy is 8, which matches the number of anyons, as we discussed before. Different ground states can be distinguished by the eigenvalue of $F^{\Gamma}_{\xi}$ operator with $\xi$ to be a non-contractable loop.

There are three irreducible representations for $S_3$, which are the trivial representation $\mathbf{1}$, the sign representation $s$, and a 2-dimensional representation $\pi$. They have the following fusion rules,

\begin{table}[h]
\centering
\resizebox{0.2\columnwidth}{!}{
\begin{tabular}{|c|c c c|}
\hline
$\otimes$ & $\mathbf{1}$ & $s$ & ${\pi}$\\ \hline
$\mathbf{1}$ & $\mathbf{1}$ & $s$ & ${\pi}$\\ 
$s$     & $s$ & $\mathbf{1}$ &${\pi}$\\ 
${\pi}$ & ${\pi}$ & ${\pi}$ & $\mathbf{1} \oplus s \oplus {\pi}$\\ \hline
\end{tabular}
}.
\caption{The fusion table of $\mathrm{Rep}(S_3)$ category.}
\end{table}
The maximal local decoherence channel is thus given by
\begin{equation}
    \begin{aligned}
        \mathcal{N}_{e} [\rho] = \frac{\rho + Z_e \rho Z_e + 2 Z_{\pi,e, \alpha, \alpha'} \rho Z_{\pi,e, \alpha, \alpha'}^{\dagger}}{6},
    \end{aligned}
\end{equation}
and $\mathcal{N} = \bigotimes_{e} \mathcal{N}_{e}$. The generalized $Z$ operators 
\begin{equation}
        \begin{aligned}
            Z_e &=  \left(F_{e}^{e,e} + F_{e}^{e,c} + F_{e}^{e,c^2} - F_{e}^{e,t} - F_{e}^{e,tc} - F_{e}^{e,tc^2}\right), \label{eq:B_anyon}\\
        Z_{\pi,e} &=  \begin{pmatrix}
        1 & 0 \\
        0 & 1 
    \end{pmatrix} F_{e}^{e,e} + \begin{pmatrix}
        \bar{\omega} & 0 \\
        0 & \omega 
    \end{pmatrix} F_{e}^{e,c} + \begin{pmatrix}
        \omega & 0 \\
        0 & \bar{\omega} 
    \end{pmatrix} F_{e}^{e,c^2} + \begin{pmatrix}
        0 & 1 \\
        1 & 0 
    \end{pmatrix} F_{e}^{e,t} + \begin{pmatrix}
        0 & \omega \\
        \bar{\omega} & 0 
    \end{pmatrix} F_{e}^{e,tc} + \begin{pmatrix}
        0 & \bar{\omega} \\
        \omega & 0 
    \end{pmatrix} F_{e}^{e,tc^2},
        \end{aligned}
    \end{equation}
where $\omega = e^{2 \pi i/3}$. The two non-trivial irreducible representations correspond to the electric charges $B$ and $C$ in the $S_3$ quantum double. Together with the vacuum $A$, these anyons form a closed subcategory and braid trivially with one another. Consequently, after incoherent proliferation, the electric 1-form symmetries remain strong, while the magnetic 1-form symmetries are reduced to weak.

Moreover, the $S_3$ quantum double exhibits a generalized $e$–$m$ duality, namely the $C$–$F$ duality~\cite{beigi2010quantum-2ab,li2024domain}, which exchanges the non-abelian electric charge $C$ with the non-abelian magnetic flux $F$. This allows for an alternative choice of Kraus operators that incoherently proliferate the $A$, $B$, and $F$ anyons. Since these anyons also form a closed subcategory and braid trivially with one another, the following symmetries remain strong after decoherence,
\begin{equation}
\begin{aligned}
W_{\xi_C}^s &= \left(F_{\xi_C}^{e,e}+F_{\xi_C}^{e,c}+F_{\xi_C}^{e,c^2}-F_{\xi_C}^{e,t}-F_{\xi_C}^{e,tc}-F_{\xi_C}^{e,tc^2}\right),\\
\mathcal{W}_{\xi_C}^{[c]}&=\frac{1}{2}\tr\left[\begin{pmatrix}F_{\xi_C}^{c^2,e}&F_{\xi_C}^{c^2,t}\\F_{\xi_C}^{c,t}&F_{\xi_C}^{c,e}\end{pmatrix}+\begin{pmatrix}F_{\xi_C}^{c^2,c}&F_{\xi_C}^{c^2,tc^2}\\F_{\xi_C}^{c,tc}&F_{\xi_C}^{c,c^2}\end{pmatrix}+\begin{pmatrix}F_{\xi_C}^{c^2,c^2}&F_{\xi_C}^{c^2,tc}\\F_{\xi_C}^{c,tc^2}&F_{\xi_C}^{c,c}\end{pmatrix}\right],
\end{aligned}
\end{equation}
whereas the non-abelian electric 1-form symmetry $\mathcal{W}_{\xi_C}^{\pi}$ and the non-abelian magnetic 1-form symmetry $\mathcal{W}_{\xi_C}^{[t]}$ are reduced to weak for any contractible closed ribbon $\xi_C$.

\section{Example: Decohering \texorpdfstring{$D_4$}{D4} quantum double}

The group of $D_4$ is the symmetry of a square. There are two generators $c$ and $t$ satisfying the following relations
\begin{align}
    c^4 = t^2 = e,\quad tct = c^3,
\end{align}
where $e$ is the identity element. According to the previous discussion, the ground states of $D_4$ quantum double on a torus are labeled as follows,
\begin{equation}
    \begin{aligned}
        &|e, e\rangle,\ |e, c\rangle,\ |e, c^2\rangle,\ |e, t\rangle,\ |e, tc\rangle \\
        &|c, e\rangle,\ |c, c\rangle,\ |c, c^2\rangle,\ |c, c^3\rangle, \\
        &|c^2, e\rangle,\ |c^2, c\rangle,\ |c^2, c^2\rangle,\ |c^2, t\rangle,\ |c^2, tc\rangle, \\
        &|t, e\rangle,\ |t, t\rangle,\ |t, c^2\rangle,\ |t, tc^2\rangle, \\
        &|tc, e\rangle,\ |tc, tc\rangle,\ |tc, tc^3\rangle,\ |tc, c^2\rangle.
    \end{aligned}
\end{equation}
In total there are 22 distinct degenerate ground states. There are 5 types of electric charges, and the fusions between them form a fusion category $\mathrm{Rep}(D_4)$. The fusion table is given as follows,
\begin{table}[h]
\centering
\resizebox{0.3\columnwidth}{!}{
\begin{tabular}{c|ccccc}
$\otimes$ & $1$ & $s_1$ & $s_2$ & $s_3$ & $\pi$ \\
\hline
$1$ & $1$ & $s_1$ & $s_2$ & $s_3$ & $\pi$ \\
$s_1$ & $s_1$ & $1$ & $s_3$ & $s_2$ & $\pi$ \\
$s_2$ & $s_2$ & $s_3$ & $1$ & $s_1$ & $\pi$ \\
$s_3$ & $s_3$ & $s_2$ & $s_1$ & $1$ & $\pi$ \\
$\pi$ & $\pi$ & $\pi$ & $\pi$ & $\pi$
& $1\oplus s_1\oplus s_2\oplus s_3$
\end{tabular}%
}
\caption{The fusion table of $\mathrm{Rep}(D_4)$ category.}
\end{table}
and they braid trivially with each other. The generalized $Z$ operators are given as follows,
\begin{equation}
\begin{aligned}
Z_{s_1,e} &=
F_e^{e,e}-F_e^{e,c}+F_e^{e,c^2}-F_e^{e,c^3}
+F_e^{e,t}-F_e^{e,tc}+F_e^{e,tc^2}-F_e^{e,tc^3}, \\
Z_{s_2,e} &= F_e^{e,e}+F_e^{e,c}+F_e^{e,c^2}+F_e^{e,c^3} -F_e^{e,t}-F_e^{e,tc}-F_e^{e,tc^2}-F_e^{e,tc^3}, \\
Z_{s_3,e} &= F_e^{e,e}-F_e^{e,c}+F_e^{e,c^2}-F_e^{e,c^3} -F_e^{e,t}+F_e^{e,tc}-F_e^{e,tc^2}+F_e^{e,tc^3}.\\
Z_{\pi,e}
&=
\begin{pmatrix}
F_e^{e,e}-iF_e^{e,c}-F_e^{e,c^2}+iF_e^{e,c^3}
&
F_e^{e,t}+iF_e^{e,tc}-F_e^{e,tc^2}-iF_e^{e,tc^3}
\\
F_e^{e,t}-iF_e^{e,tc}-F_e^{e,tc^2}+iF_e^{e,tc^3}
&
F_e^{e,e}+iF_e^{e,c}-F_e^{e,c^2}-iF_e^{e,c^3}
\end{pmatrix}
\end{aligned}
\end{equation}
Therefore, the maximal decoherence channel is given by $\mathcal{N} = \bigotimes_{e} \mathcal{N}_e$, where
\begin{equation}
    \begin{aligned}
        \mathcal{N}_{e}[\rho] = \frac{\rho + \sum_{i=1}^{3} Z_{s_i,e} \rho Z_{s_i,e}^{\dagger} + 2 Z_{\pi, e, \alpha, \alpha'} \rho Z_{\pi, e, \alpha, \alpha'}^{\dagger}}{8}
    \end{aligned}
\end{equation}

\end{document}